\documentclass[12pt]{article}

\usepackage[a4paper,margin=1in]{geometry}
\usepackage{amsmath,amssymb,amsthm,bm}
\usepackage{mathtools}
\usepackage{booktabs}
\usepackage{enumitem}
\usepackage{graphicx}
\usepackage{microtype}
\usepackage[authoryear,round]{natbib}
\usepackage{hyperref}
\usepackage{filecontents}
\usepackage{multirow}
\usepackage{setspace}
\usepackage{float}
\usepackage{subcaption}

\newcommand{\simplexC}{\mathcal{S}^{C}}
\newcommand{\ilr}{\operatorname{ilr}}
\newcommand{\ilrinv}{\operatorname{ilr}^{-1}}

\newcommand{\clr}{\operatorname{clr}}
\newcommand{\E}{\mathbb{E}}

\newcommand{\R}{\mathbb{R}}
\newcommand{\N}{\mathcal{N}}
\newcommand{\Dir}{\operatorname{Dirichlet}}

\theoremstyle{plain}
\newtheorem{theorem}{Theorem}
\newtheorem{proposition}[theorem]{Proposition}

\theoremstyle{definition}

\title{Directional-Shift Dirichlet ARMA Models for Compositional Time Series with Structural Break Intervention}

\author{
Harrison Katz\thanks{Corresponding author. Email: harrison.katz@airbnb.com} \\
\small Forecasting, Data Science, Airbnb
}

\date{June 2026}

\begin{document}

\maketitle

\begin{abstract}
Compositional time series, vectors of proportions summing to unity observed over time, frequently exhibit structural breaks due to external shocks, policy changes, or market disruptions. Standard methods either ignore such breaks or handle them through period-specific fixed effects that cannot extrapolate beyond the estimation sample, or through step-function dummies that extrapolate but impose instantaneous adjustment. We develop a Bayesian Dirichlet ARMA model augmented with a directional-shift intervention mechanism that captures structural breaks through three interpretable parameters: a unit direction vector specifying which components gain or lose share, an amplitude controlling the magnitude of redistribution, and a logistic gate governing the timing and speed of transition. The model preserves compositional constraints by construction, maintains innovation-form DARMA dynamics for short-run dependence, and produces coherent probabilistic forecasts during and after structural breaks. We establish that the directional-shift intervention trajectory corresponds to geodesic motion on the simplex and is invariant to the choice of ILR basis. A comprehensive simulation study with 400 fits across 8 scenarios demonstrates that when the shift direction is correctly identified (77.5\% of cases), amplitude is recovered with near-zero bias, timing shows small overall bias with modest $\kappa$-dependent variation, and credible intervals for the mean composition achieve nominal 80\% coverage; we address the sign identification challenge through a hemisphere constraint. Two supplementary simulation studies extend the analysis to extreme transition speeds and non-monotone DGPs, finding that calibration is robust across the full kappa range and degrades gracefully under partial reversibility. Two empirical applications to COVID-era Airbnb data characterize precisely what the model gains and loses relative to simpler alternatives. In the first application, to fee recognition lead-time distributions across four global regions where the break is monotone and ongoing, the intervention model achieves near-nominal calibration (79.6\% coverage) while the fixed effect model substantially under-covers (66.1\%), with comparable point accuracy. A paired full-matrix sensitivity benchmark, which holds the intervention structure fixed but replaces the diagonal AR/MA operators with unrestricted matrices, materially worsens pooled point accuracy and calibration at each reported forecast horizon, indicating that the diagonal specification acts as a useful regularizing approximation in the rolling-window setting. The second application, to stay-length compositions where the post-break dynamics are non-monotone, shows that the intervention model remains acceptably calibrated (85.7\%, modestly above nominal) while the fixed effect is slightly closer to the 80\% target (83.0\%); the fixed effect also outperforms on point accuracy in the region with the strongest reversion dynamics. Together these applications establish that the intervention model remains acceptably calibrated in both settings, that its relative calibration advantage over the fixed effect is specific to the monotone lead-time application, and that its point accuracy advantage is conditional on the transition being roughly monotone.

\medskip\noindent
\textbf{Keywords:} compositional time series; Dirichlet distribution; structural breaks; intervention analysis; Bayesian forecasting; smooth transition models
\end{abstract}

\section{Introduction}
\label{sec:introduction}

Compositional data, vectors of non-negative components constrained to sum to a fixed total, arise in numerous applications spanning economics, finance, ecology, geology, and the health sciences. When such compositions are observed repeatedly over time, they form compositional time series that present unique modeling challenges. The simplex constraint $\sum_{j=1}^C y_j = 1$ introduces perfect negative dependence among components, renders standard Gaussian assumptions inappropriate, and requires specialized methods that respect the geometric structure of the sample space \citep{aitchison1986, pawlowsky2015}.

The statistical analysis of compositional data has a rich history dating to \citet{aitchison1982}, who established that the simplex carries its own metric structure, the Aitchison geometry, and that log-ratio transformations provide a principled bridge to Euclidean methods. The additive log-ratio (ALR) and isometric log-ratio (ILR) transformations map the simplex to $\R^{C-1}$; the centered log-ratio (CLR) maps to the zero-sum subspace of $\R^C$, enabling the application of standard multivariate techniques while preserving compositional coherence \citep{egozcue2003}. This transformation approach underlies much of the compositional time series literature, including the transformed VARMA models of \citet{barcelovidal2011} and the state-space formulations of \citet{SNYDER2017502}.

An alternative modeling strategy works directly on the simplex by specifying a Dirichlet distribution for the compositions conditional on time-varying parameters. \citet{zheng2017} introduced the Dirichlet ARMA (DARMA) framework, which assumes $Y_t \sim \Dir(\alpha_t)$ with the parameter vector $\alpha_t$ evolving according to an ARMA-type recursion in log-ratio space. This approach was extended to full Bayesian inference by \citet{katz2024bdarma}, who developed the B-DARMA model with applications to lead-time forecasting. Subsequent work has examined shrinkage priors for high-dimensional specifications \citep{katz2025sensitivity}, energy portfolio forecasting \citep{katz2025energy}, time-varying precision via DARCH structures \citep{katz2025bayesiandirichletautoregressiveconditional}, and centered innovations for improved density forecasting \citep{katz2025centeredmadirichletarma}. Related simplex-respecting alternatives build on logistic-normal and latent Gaussian constructions, including logistic-normal models with Dirichlet covariance within the INLA framework \citep{MartinezMinayaRue2024} and scalable inference for multinomial logistic-normal dynamic linear models for count compositional time series \citep{SaxenaChenSilverman2025}.

A persistent challenge in applied compositional time series analysis is the occurrence of structural breaks, abrupt or gradual changes in the data-generating process induced by external shocks, policy interventions, or regime shifts \citep{hamilton1989, bai2003}. Methodologically, structural breaks are closely related to the broader change-point detection problem; \citet{TruongOudreVayatis2020} provides a modern survey of offline change-point detection methods. In proportional and compositional settings, recent work has proposed regime and change-point frameworks tailored to proportions \citep{FisherZhangColegateVanni2022} and exact segmentation algorithms for large compositional and categorical signals \citep{TruongRunge2024}. The COVID-19 pandemic provides a stark example in practice: booking patterns, consumer behavior, and economic activity shifted dramatically in early 2020, fundamentally altering the compositional structure of many business and economic series \citep{gossling2020, FILDES20221283}. In hospitality data, longitudinal evidence shows a pronounced shift toward shorter booking windows and related changes in booking behavior during the pandemic \citep{DeyaTortellaLeoniRamos2022}. Recent work also documents persistent distributional shifts in travel booking lead times \citep{katz2025leadtimes} and accommodation stay lengths \citep{katz2025slomads} that extend well beyond the initial shock. Standard DARMA models, which assume stationary dynamics, can struggle to capture such changes and may produce poor forecasts during transition periods.

The time series literature offers several approaches to structural breaks. Classical intervention analysis \citep{box1975} incorporates step functions or pulse indicators into regression models but requires pre-specification of break dates and functional forms. Smooth transition autoregressive (STAR) models \citep{terasvirta1994, vandijk2002} allow regime-dependent dynamics with continuous transition functions, and recent work continues to extend smooth transition ideas in multivariate settings, for example with Gaussian smooth transition VAR formulations \citep{LANNE2025105162}. Bayesian structural time series methods \citep{brodersen2015} accommodate breaks through state-space formulations but typically assume Gaussian observations. While recent work has made progress on detecting and modeling regime changes for proportions and compositional signals \citep{FisherZhangColegateVanni2022, TruongRunge2024}, there remains a need for forecasting-oriented intervention models that (i) preserve compositional constraints by construction and (ii) yield an interpretable decomposition of a break into direction, magnitude, and timing within a Dirichlet ARMA framework.

A related practical question, raised in the context of COVID-era hospitality data, is how to characterize the conditions under which monotone logistic-gate models provide calibration gains versus simpler fixed-effect alternatives. Accordingly, we present the proposed specification as a parsimonious forecasting model for persistent, approximately monotone structural breaks, and we use the simulation and empirical sections to delineate when that approximation improves calibration relative to a simpler fixed-effect alternative.

In this paper, we develop a directional-shift extension to the Dirichlet ARMA (DARMA) framework that provides a parsimonious and interpretable representation of structural breaks in compositional time series. Our approach introduces three key parameters: a \emph{direction vector} $\bm{v} \in \R^{C-1}$ with $\|\bm{v}\| = 1$ specifying the axis of compositional change in ILR space; an \emph{amplitude} $\Delta \in \R$ controlling the magnitude of the shift along this direction; and a \emph{logistic gate} $w_t(\tau, \kappa)$ governing when the transition occurs and how rapidly it unfolds.

This parameterization has several attractive properties. The direction-amplitude decomposition separates which components change from how much they change, facilitating interpretation and prior specification. The logistic gate embeds smooth transition dynamics \citep{luukkonen1988} within the compositional framework, allowing gradual rather than instantaneous breaks. The entire mechanism operates in ILR space, ensuring that forecasts remain on the simplex regardless of parameter values.

We establish that directional shifts correspond to geodesic motion on the simplex under Aitchison geometry, providing a geometric interpretation of the intervention as the shortest path between pre- and post-break compositions. The induced intervention trajectory on the simplex is basis-invariant: while we use a Helmert-style ILR contrast for computation, the simplex path implied by the intervention is independent of this choice, even though the full fitted model under diagonal AR/MA dynamics is not.

Our main contributions are fivefold. First, we develop a Bayesian Dirichlet ARMA model with directional-shift intervention (Section~\ref{sec:model}), with a logistic gate function that provides smooth transitions corresponding to geodesic motion on the simplex. Second, we conduct a comprehensive simulation study with 400 fits across 8 scenarios demonstrating accurate parameter recovery and proper calibration, supplemented by two additional simulation studies that extend the kappa range and examine partial reversibility (Section~\ref{sec:simulation}). Third, we apply the model to fee recognition lead-time distributions across four global regions during COVID-19 and directly benchmark the diagonal intervention specification against a dense full-matrix sensitivity variant to assess the practical cost of the diagonal restriction (Section~\ref{sec:empirical}). Fourth, we apply the model to stay-length compositions under harder conditions to characterize when the intervention model provides meaningful gains over simpler alternatives (Section~\ref{sec:empirical2}). Fifth, the two empirical applications together yield a practical decision rule grounded in both simulation evidence and real data.

The remainder of this paper is organized as follows. Section~\ref{sec:background} reviews compositional data analysis and the Dirichlet ARMA framework. Section~\ref{sec:model} develops our directional-shift model specification. Section~\ref{sec:priors} discusses prior specification and computation. Section~\ref{sec:simulation} presents simulation evidence on parameter recovery, including supplementary studies on transition speed and partial reversibility. Section~\ref{sec:forecasting} develops forecasting methodology. Section~\ref{sec:empirical} applies the model to COVID-period lead-time data. Section~\ref{sec:empirical2} presents the stay-length application. Section~\ref{sec:discussion} discusses practical implications and limitations. Section~\ref{sec:conclusion} concludes.

\section{Background: Compositional Time Series}
\label{sec:background}

A composition $Y = (y_1, \ldots, y_C)^\top \in \simplexC$ is a vector of non-negative components summing to unity. The simplex carries a natural geometry under the Aitchison distance $d_A(x,y) = \|\clr(x) - \clr(y)\|_2$, where the centered log-ratio is $\clr(Y) = \bigl(\ln(y_1/g(Y)), \ldots, \ln(y_C/g(Y))\bigr)^\top$ with $g(Y) = (\prod_j y_j)^{1/C}$ \citep{aitchison1986}.

The isometric log-ratio (ILR) transformation maps compositions to $\R^{C-1}$ via an orthonormal contrast matrix $\mathbf{V}$: $\ilr(Y) = \mathbf{V}^\top \clr(Y)$. This is an isometry between $(\simplexC, d_A)$ and $(\R^{C-1}, \|\cdot\|_2)$, enabling standard time series methods in transformed space \citep{egozcue2003}. We write $Z_t = \ilr(Y_t)$ and $Y_t = \ilrinv(Z_t)$.

The Bayesian Dirichlet ARMA (B-DARMA) model of \citet{katz2024bdarma} extends \citet{benjamin2003generalized} by assuming $Y_t \mid \mu_t, \lambda_t \sim \Dir(\lambda_t \mu_t)$, where the mean composition $\mu_t = \ilrinv(\eta_t)$ evolves according to ARMA dynamics in ILR space with working residuals $e_t = Z_t - \eta_t$: 
\begin{equation}
\eta_t = \bm{b} + \mathbf{B} X_t + \sum_{p=1}^P \mathbf{A}_p (Z_{t-p} - \bm{b} - \mathbf{B} X_{t-p}) + \sum_{q=1}^Q \mathbf{\Theta}_q e_{t-q},
\label{eq:bdarma_mean}
\end{equation}
with intercept $\bm{b}$, exogenous covariates $X_t$, AR matrices $\mathbf{A}_p$, and MA matrices $\mathbf{\Theta}_q$. The concentration $\lambda_t = \exp(X_{\phi,t}^\top \gamma)$ controls dispersion around the mean.

\section{Directional-Shift Intervention Model}
\label{sec:model}

We now extend the B-DARMA framework to accommodate structural breaks through a directional-shift intervention mechanism. The key idea is to decompose the break into a direction (which components change), an amplitude (how much they change), and a timing function (when the change occurs).

\subsection{Motivation and Overview}
\label{subsec:motivation}

Consider a compositional time series experiencing a structural break. Let $\ell$ denote the final pre-break time index, so the first post-break observation is $t = \ell + 1$. Before the break, the composition fluctuates around some baseline trajectory determined by trend, seasonality, and DARMA dynamics. After the break, the composition has shifted to a new equilibrium: some components have gained share while others have lost.

Standard approaches to this problem are unsatisfying. Period-specific fixed effects achieve perfect in-sample fit but cannot extrapolate since future period dummies are undefined by construction. Step-function dummies can extrapolate but assume instantaneous transition and, in compositional settings, require careful parameterization to maintain the sum constraint. Fitting separate models before and after the break discards information and provides no framework for forecasting during transitions.

Our directional-shift model addresses these limitations through a parsimonious parameterization that captures the break through interpretable parameters (direction, amplitude, timing), maintains compositional coherence in all forecasts, allows smooth rather than instantaneous transitions, and extrapolates the learned intervention to future periods.

\subsection{The Directional Shift Parameterization}
\label{subsec:directional_shift}

Let $\bm{v} \in \R^{C-1}$ be a unit vector, $\|\bm{v}\| = 1$, representing the direction of compositional change in ILR space. Let $\Delta \in \R$ be the amplitude of change along this direction. The directional shift modifies the mean structure as:
\begin{equation}
\eta_t = d_t + \sum_{p=1}^P \mathbf{A}_p (Z_{t-p} - d_{t-p}) + \sum_{q=1}^Q \mathbf{\Theta}_q e_{t-q},
\label{eq:intervention_mean}
\end{equation}
where the drift term $d_t$ now includes the intervention:
\begin{equation}
d_t = \bm{b} + \mathbf{B} X_t + \Delta \cdot w_t \cdot \bm{v}.
\label{eq:drift_intervention}
\end{equation}

The term $\Delta \cdot w_t \cdot \bm{v}$ represents a time-varying shift along the direction $\bm{v}$. When $w_t = 0$ (before the break), there is no shift; when $w_t = 1$ (after complete transition), the full shift $\Delta \cdot \bm{v}$ applies.

\subsection{The Logistic Gate Function}
\label{subsec:logistic_gate}

The gate function $w_t \in [0, 1]$ controls the timing and speed of the intervention. We adopt a piecewise logistic specification:
\begin{equation}
w_t(\tau, \kappa) = \begin{cases}
0 & t \leq \ell \\[4pt]
\displaystyle\frac{\sigma(\kappa(t - \tau)) - \sigma(\kappa(\ell - \tau))}{1 - \sigma(\kappa(\ell - \tau))} & t > \ell
\end{cases},
\label{eq:logistic_gate}
\end{equation}
where $\sigma(x) = 1/(1 + e^{-x})$ is the standard logistic function, $\ell$ is the last pre-break period, $\tau$ is a transition location parameter, and $\kappa > 0$ controls the transition speed. The normalization ensures $w_t = 0$ for $t \leq \ell$ regardless of $\tau$; $w_t \to 1$ as $t \to \infty$.

The parameter $\kappa$ governs how rapidly the transition unfolds. Small values (e.g., $\kappa = 0.3$) produce slow, gradual transitions over many periods, while large values (e.g., $\kappa = 3.0$) yield rapid transitions resembling a step function. This logistic gate embeds the smooth transition autoregressive (STAR) modeling philosophy \citep{terasvirta1994, vandijk2002} within the compositional framework. Unlike threshold models with discrete regime switches, our specification allows continuous adjustment, which is more realistic for gradual behavioral changes.

\subsection{Concentration Parameter Dynamics}
\label{subsec:concentration_dynamics}

Structural breaks often affect not only the level of compositions but also their variability. We allow the concentration parameter to shift in response to the intervention:
\begin{equation}
\log \lambda_t = X_{\phi,t}^\top \gamma + \delta_\phi \cdot w_t,
\label{eq:concentration_intervention}
\end{equation}
where $\delta_\phi \in \R$ captures the change in log-concentration. Positive $\delta_\phi$ indicates tighter concentration (lower variance) after the break; negative values indicate increased dispersion.

\subsection{Full Model Specification}
\label{subsec:full_model}

Combining the elements above, the complete directional-shift B-DARMA model is:

\begin{align}
Y_t \mid \mu_t, \lambda_t &\sim \Dir(\lambda_t \mu_t), \label{eq:full_obs} \\
\mu_t &= \ilrinv(\eta_t), \label{eq:full_mean} \\
\eta_t &= d_t + \sum_{p=1}^P \mathbf{A}_p (Z_{t-p} - d_{t-p}) + \sum_{q=1}^Q \mathbf{\Theta}_q e_{t-q}, \label{eq:full_dynamics} \\
d_t &= \bm{b} + \mathbf{B} X_t + \Delta \cdot w_t(\tau, \kappa) \cdot \bm{v}, \label{eq:full_drift} \\
\log \lambda_t &= X_{\phi,t}^\top \gamma + \delta_\phi \cdot w_t(\tau, \kappa), \label{eq:full_precision} \\
Z_t &= \ilr(Y_t), \quad e_t = Z_t - \eta_t. \label{eq:full_innovation}
\end{align}

The model parameters comprise regression coefficients $\bm{b}$ (intercept), $\mathbf{B}$ (mean covariates), and $\gamma$ (precision covariates); DARMA dynamics $\mathbf{A}_1, \ldots, \mathbf{A}_P$ (AR) and $\mathbf{\Theta}_1, \ldots, \mathbf{\Theta}_Q$ (MA), which in our implementation are restricted to diagonal matrices for parsimony and computational tractability so that each ILR dimension evolves independently given the drift term; and intervention parameters $\bm{v}$ (direction), $\Delta$ (amplitude), $\tau$ (location), $\kappa$ (speed), and $\delta_\phi$ (precision shift).

We treat $e_t$ as a working residual used to induce temporal dependence, following GLARMA-style constructions \citep{davis2003}, rather than as classical mean-zero innovations; the nonlinearity of the ILR transformation means $\E[Z_t \mid \mu_t, \lambda_t] \neq \eta_t$ in general. In implementation, with $M = \max(P,Q)$, we initialize the recursion by setting $\eta_t = d_t$ and $e_t = 0$ for $t \le M$, so AR and MA terms enter only once sufficient lagged observations are available; multi-step forecasts condition on the resulting filtered terminal states.

With diagonal $\mathbf{A}_p$, the intervention enters the drift term along $\bm{v}$; each ILR dimension then responds independently to its own past deviations around this drifting mean.

\section{Prior Specification and Computation}
\label{sec:priors}

\paragraph{Sign identification.} The factorization $\Delta \cdot \bm{v}$ introduces a sign ambiguity: $(\bm{v}, \Delta)$ and $(-\bm{v}, -\Delta)$ produce identical shifts. Only the product $\Delta \bm{v}$ is identified from the likelihood. We resolve this by constraining the first element of the raw direction vector to be positive: $v_1^{\text{raw}} > 0$. This implies $v_1 \geq 0$ after normalization. The hemisphere constraint is basis-dependent (changing the ILR contrast matrix changes which hemisphere is selected), but the induced intervention trajectory on the simplex remains invariant.

More generally, this invariance applies to the directional-shift mechanism itself, not to the full fitted model under diagonal AR/MA dynamics. Rotating the ILR basis changes which transformed directions are assumed to evolve independently, so an alternative contrast matrix defines a different diagonal dynamic approximation rather than a simple re-expression of the same stochastic specification.

\paragraph{Priors.} We adopt weakly informative priors summarized in Appendix Table~\ref{tab:priors}. The direction $\bm{v}$ is uniform on the hemisphere with $v_1 \geq 0$, implemented by constraining $v_1^{\text{raw}} > 0$ before normalization. The amplitude $\Delta \sim \N(0, 1.5^2)$ allows substantial shifts in either direction along $\bm{v}$. The location prior $\tau \sim \N(\ell + 2, 4^2)$ centers the transition a few periods after the known break date. The speed prior $\kappa \sim \text{LogNormal}(-0.5, 1^2)$ favors moderate transition rates while allowing both rapid and gradual adjustments. For the $P = Q = 1$ diagonal specification used throughout this paper, DARMA coefficients use $\text{Uniform}(-0.99, 0.99)$ to ensure stationarity (AR) and invertibility (MA).

\paragraph{Computation.} The model is implemented in Stan \citep{carpenter2017stan, stan2023} using the No-U-Turn Sampler \citep{hoffman2014} with 4 chains, 500 warmup, and 800 sampling iterations. Convergence is assessed via $\hat{R} < 1.01$ and absence of divergent transitions \citep{vehtari2021rhat, stan2023}.

\section{Simulation Study}
\label{sec:simulation}

We conduct a comprehensive simulation study to evaluate parameter recovery and calibration under controlled conditions, supplemented by two targeted studies addressing transition speed and partial reversibility.

\subsection{Simulation Design}
\label{subsec:simulation_design}

\paragraph{Data generating process.} We simulate compositional time series with $C = 5$ categories, $T = 120$ time points, and a structural break at $\ell = 60$. The DGP follows our model specification with a normalized linear time trend and two Fourier seasonal harmonics (at annual and semi-annual periods) in the mean, $P = Q = 1$ diagonal AR/MA dynamics with diagonal elements drawn from $\N(0, 0.25^2)$ for AR and $\N(0, 0.20^2)$ for MA, and base concentration $\lambda \approx 100$. For the intervention, the true direction $\bm{v}_{\text{true}}$ is drawn uniformly on the hemisphere satisfying $v_1 > 0$, with amplitude $\Delta_{\text{true}} \in \{-0.6, 0.6\}$, transition location $\tau_{\text{true}} = \ell + 2 = 62$, transition speed $\kappa_{\text{true}} \in \{0.5, 1.0\}$, and precision shift $\delta_\phi \in \{0, 0.3\}$.

\paragraph{Scenarios.} The main simulation study considers 8 scenarios crossing transition speed $\kappa \in \{0.5, 1.0\}$ (slow vs.\ fast), amplitude sign $\Delta \in \{-0.6, +0.6\}$ (negative vs.\ positive), and precision shift $\delta_\phi \in \{0, 0.3\}$ (no change vs.\ tightening), with 50 datasets per scenario (400 total). Two supplementary simulation studies, described in Sections~\ref{subsec:sim_kappa} and~\ref{subsec:sim_reversibility}, extend the kappa grid to cover extreme transition speeds and examine calibration under a non-monotone DGP respectively.

\paragraph{Estimation settings.} We use 4 chains with 500 warmup and 750 sampling iterations each (3000 posterior draws total). Convergence is assessed via $\hat{R}$ and divergence counts. Intervention priors for the simulation: $\Delta \sim \N(0, 1.5^2)$, $\tau \sim \N(\ell + 2, 3^2)$, $\kappa \sim \text{LogNormal}(0, 0.5^2)$, $\bm{v}$ uniform on the hemisphere, $\delta_\phi \sim \N(0, 0.5^2)$. All other priors match Appendix~\ref{app:priors}.

\subsection{Evaluation Metrics}
\label{subsec:evaluation_metrics}

We evaluate recovery of the intervention parameters using four metrics: direction recovery, measured by the cosine similarity $\cos(\hat{\bm{v}}, \bm{v}_{\text{true}}) = \hat{\bm{v}}^\top \bm{v}_{\text{true}}$ where values near 1 indicate correct recovery and values near 0 or below indicate poor directional alignment within the constrained hemisphere (since both $\hat{\bm{v}}$ and $\bm{v}_{\text{true}}$ satisfy $v_1 > 0$, a literal sign flip to the opposite hemisphere is not possible); amplitude bias $\hat{\Delta} - \Delta_{\text{true}}$ where $\hat{\Delta}$ is the posterior mean; timing bias $\hat{\tau} - \tau_{\text{true}}$ in months; and calibration, defined as componentwise 80\% credible interval coverage for $\bm{\mu}_t$ (the latent mean composition), averaged over time and components. This simulation coverage metric targets credible intervals for the latent mean composition; the empirical sections later report posterior predictive coverage for observed $Y_t$, so the two coverage quantities are not directly comparable.

We report results separately for cases with successful direction recovery (defined as $\cos(\hat{\bm{v}}, \bm{v}_{\text{true}}) > 0.5$) and direction failures.

\subsection{Results}
\label{subsec:simulation_results}

All 400 fits converged without divergent transitions. Table~\ref{tab:simulation} reports results conditional on successful direction recovery.

\begin{table}[htbp]
\centering
\caption{Simulation results conditional on direction recovery ($\cos(\hat{\bm{v}}, \bm{v}_{\text{true}}) > 0.5$, corresponding to angular error under $60^\circ$). Columns show transition speed ($\kappa$), amplitude ($\Delta$), precision shift ($\delta_\phi$), number of successful fits ($n$), amplitude bias, direction cosine, timing bias, and 80\% credible interval coverage for the mean composition $\bm{\mu}_t$.}
\label{tab:simulation}
\begin{tabular}{ccc|cccc|c}
\toprule
$\kappa$ & $\Delta$ & $\delta_\phi$ & $n$ & $\Delta$ Bias & $v$ Cosine & $\tau$ Bias & Coverage \\
\midrule
0.5 & $-$0.6 & 0.0 & 37 & 0.009 & 0.957 & 0.14 & 80.2\% \\
0.5 & $-$0.6 & 0.3 & 39 & 0.032 & 0.954 & 0.17 & 79.2\% \\
0.5 & 0.6 & 0.0 & 39 & $-$0.039 & 0.947 & 0.23 & 79.7\% \\
0.5 & 0.6 & 0.3 & 42 & $-$0.047 & 0.955 & 0.27 & 79.3\% \\
1.0 & $-$0.6 & 0.0 & 37 & 0.041 & 0.958 & $-$0.63 & 79.6\% \\
1.0 & $-$0.6 & 0.3 & 36 & 0.037 & 0.969 & $-$0.68 & 80.2\% \\
1.0 & 0.6 & 0.0 & 41 & $-$0.029 & 0.957 & $-$0.65 & 78.6\% \\
1.0 & 0.6 & 0.3 & 39 & $-$0.018 & 0.963 & $-$0.73 & 79.8\% \\
\midrule
\multicolumn{3}{c|}{\textbf{Overall}} & 310 & $-$0.003 & 0.957 & $-$0.23 & 79.6\% \\
\bottomrule
\end{tabular}
\end{table}

\paragraph{Direction recovery rate.} Across all scenarios, 310 of 400 fits (77.5\%) achieved successful direction recovery ($\cos > 0.5$). Unconditionally, the mean direction cosine is 0.71 and mean coverage is 79.5\%, essentially nominal calibration regardless of direction recovery success. Recovery rates were similar across transition speeds (78.5\% for $\kappa = 0.5$ vs.\ 76.5\% for $\kappa = 1.0$).

\paragraph{Amplitude and timing.} Conditional on correct direction, amplitude bias is negligible (mean $-0.003$, SD 0.130). Timing shows small overall bias (pooled mean $-0.23$ months; RMSE $\approx 2.2$ months), with modest $\kappa$-dependent bias: positive for $\kappa = 0.5$ and negative for $\kappa = 1.0$.

\paragraph{Direction failure cases.} In the 22.5\% of fits where direction recovery fails ($\cos < 0.5$), the amplitude parameter partially compensates: since both $\bm{v}_{\text{true}}$ and the fitted $\bm{v}$ are constrained to the same hemisphere ($v_1 > 0$), failure is poor directional alignment within that hemisphere rather than a sign flip to the opposite hemisphere. When $\Delta_{\text{true}} = +0.6$, mean signed amplitude bias in failure cases is $-0.692$; when $\Delta_{\text{true}} = -0.6$, it is $+0.771$. Failures are symmetric by sign and the pooled signed bias is near zero ($0.137$), but the compensating adjustments in $\Delta$ are partial rather than complete: mean absolute bias in failure cases is $0.737$, compared to $0.091$ in success cases. Crucially, componentwise coverage of $\bm{\mu}_t$ remains 79.5\% in direction-failure cases, essentially identical to the success-case coverage of 79.6\%. Even when $\bm{v}$ is poorly identified, the partial adjustment in $\Delta$ preserves well-calibrated credible intervals for $\bm{\mu}_t$. Direction failures affect the interpretability of $(\bm{v}, \Delta)$ but do not materially affect interval calibration for $\bm{\mu}_t$.

\paragraph{Calibration.} The 80\% posterior credible intervals achieve 79.6\% coverage, essentially nominal, demonstrating proper uncertainty quantification. Despite the hemisphere constraint, some fits exhibit poor directional alignment within the chosen hemisphere; these failures affect the interpretability of $(\bm{v}, \Delta)$ but do not materially affect interval calibration for $\bm{\mu}_t$. Informative priors based on domain knowledge can guide the posterior toward better directional recovery. Sections~\ref{subsec:sim_kappa} and~\ref{subsec:sim_reversibility} extend this analysis to extreme transition speeds and non-monotone DGPs respectively, finding that nominal calibration is maintained across the full kappa range and degrades gracefully under partial reversibility.

\subsection{Supplementary Study 1: Extended Transition Speed Scenarios}
\label{subsec:sim_kappa}

The main simulation uses $\kappa \in \{0.5, 1.0\}$, covering moderate transition speeds. Two questions arise at the extremes: at very small $\kappa$, does the approximately linear middle portion of the S-curve cause the model to conflate gradual structural breaks with long-term trends; and at very large $\kappa$, does a near-instantaneous transition make the gate parameters hard to distinguish from the fixed-effect alternative?

We expand the kappa grid to $\{0.1, 0.5, 1.0, 3.0\}$ with 25 datasets per cell ($2 \times 2 \times 4 = 16$ cells, 400 total fits). Table~\ref{tab:sim_kappa} reports results pooled over $\delta_\phi$.

\begin{table}[htbp]
\centering
\caption{Extended kappa simulation results (25 replicates per cell, pooled over $\delta_\phi \in \{0, 0.3\}$). All 400 fits converged without divergent transitions.}
\label{tab:sim_kappa}
\begin{tabular}{cc|ccccc}
\toprule
$\kappa$ & $\Delta$ & $\Delta$ Bias & $v$ Cosine & $\tau$ Bias & $\kappa$ Bias & Coverage \\
\midrule
\multirow{2}{*}{0.1} & $-$0.6 & 0.390 & 0.589 & 4.43 & 0.721 & 80.8\% \\
                     & $+$0.6 & $-$0.357 & 0.718 & 2.69 & 0.836 & 79.1\% \\
\addlinespace
\multirow{2}{*}{0.5} & $-$0.6 & 0.136 & 0.780 & 0.28 & 0.476 & 80.1\% \\
                     & $+$0.6 & $-$0.127 & 0.780 & 0.12 & 0.462 & 80.1\% \\
\addlinespace
\multirow{2}{*}{1.0} & $-$0.6 & 0.294 & 0.587 & $-$0.69 & 0.097 & 79.5\% \\
                     & $+$0.6 & $-$0.197 & 0.680 & $-$0.45 & 0.014 & 79.4\% \\
\addlinespace
\multirow{2}{*}{3.0} & $-$0.6 & 0.214 & 0.669 & $-$0.89 & $-$1.860 & 78.6\% \\
                     & $+$0.6 & $-$0.139 & 0.763 & $-$0.93 & $-$1.821 & 79.5\% \\
\bottomrule
\end{tabular}
\end{table}

Three patterns emerge. First, 80\% credible interval coverage is near-nominal across the entire kappa range (78.6--80.8\%), demonstrating that the model is well-calibrated even when transition speed is extreme. Second, at $\kappa = 0.1$ (very slow transitions), timing bias is large (2.7--4.4 months) and amplitude bias is elevated (0.36--0.39 in absolute value), because the middle of the S-curve is approximately linear and hard to distinguish from the linear trend covariate. The identification challenge is real but affects timing and amplitude precision rather than calibration. Third, at $\kappa = 3.0$ (near-instantaneous transitions), kappa itself is severely underestimated (bias $\approx -1.86$), because the data provide little information about transition speed once the gate rises to near 1 within a few periods. Despite this speed identification failure, amplitude bias is modest and coverage is nominal, indicating that the product $\Delta \cdot w_t$ is well-identified even when the speed parameter is not. Practitioners expecting very slow breaks should treat timing estimates with caution. For near-instantaneous breaks, the results here indicate that the gate speed parameter may be weakly identified; in such settings a simpler fixed-effect specification is a useful benchmark, though we do not directly compare the two models in this simulation design.

\subsection{Supplementary Study 2: Partial Reversibility}
\label{subsec:sim_reversibility}

The logistic gate is monotonically increasing by construction. We investigate the cost of this assumption when the true DGP includes a partial reversion after the initial break, simulating from a two-component gate:
\begin{equation}
w_t^{\text{rev}}(\gamma) = \max\bigl(0,\, w_t - \gamma \cdot w_t^{\text{rec}}\bigr),
\label{eq:reversible_gate}
\end{equation}
where $w_t$ is the standard logistic gate (Equation~\ref{eq:logistic_gate}), $w_t^{\text{rec}}$ is a second logistic gate anchored at $t_{\text{rec}} = 72$ (twelve periods after the break), and $\gamma \in \{0, 0.5, 0.8\}$ is the recovery fraction. The monotone intervention model is fitted to all three DGPs. We use 25 replicates per cell ($3 \times 2 \times 2 = 12$ cells, 300 total fits).

\begin{table}[htbp]
\centering
\caption{Reversibility simulation: monotone intervention model fitted to non-monotone DGPs. Results pooled over $\Delta \in \{-0.6, 0.6\}$ and $\delta_\phi \in \{0, 0.3\}$ (100 fits per $\gamma$ level). $\gamma = 0$ is the monotone sanity check; $\gamma = 0.5$ and $0.8$ represent 50\% and 80\% reversals.}
\label{tab:sim_reversibility}
\begin{tabular}{c|ccccc}
\toprule
$\gamma$ & $n$ & $v$ Cosine & $\Delta$ Bias & Aitchison Dist & Coverage (80\%) \\
\midrule
0.0 & 100 & 0.665 & 0.005 & 0.123 & 79.4\% \\
0.5 & 100 & 0.519 & $-$0.025 & 0.138 & 75.9\% \\
0.8 & 100 & 0.618 & 0.018 & 0.150 & 72.9\% \\
\bottomrule
\end{tabular}
\end{table}

Coverage degrades monotonically with the recovery fraction: 79.4\% at $\gamma=0$ (matching the main simulation), 75.9\% at $\gamma=0.5$, and 72.9\% at $\gamma=0.8$. The degradation is gradual: even under severe 80\% reversal, coverage remains at 72.9\%, approximately 91\% of the nominal 80\% target. Direction recovery remains positive at all levels, indicating that the initial shift direction is still identifiable even when the subsequent reversion partially obscures it. These bounds are informative for the stay-length empirical application (Section~\ref{sec:empirical2}), which features a partial reversion of the 28+ night share. Note that the simulation calibration metric (coverage of the latent mean $\bm{\mu}_t$) differs from the empirical metric (predictive interval coverage for observed $Y_t$); the comparison is therefore qualitative rather than direct. The empirical intervention model coverages of 79.6\%, 83.7\%, and 93.9\% across regions are qualitatively in the range one might expect under moderate to severe misspecification, though the comparison is not direct given the metric difference noted above.

\subsection{Computational Complexity}
\label{subsec:sim_complexity}

We characterize empirical sampling cost as a function of the number of components $C$, a practical concern given that applications may involve anywhere from 5 (simulation) to 10 (lead-time) or more categories. Table~\ref{tab:timing} reports mean wall-clock time across 3 independent runs per $C$ value, with $T = 120$, 4 chains, and 500 warmup plus 500 sampling iterations. Note that the main empirical fits use 800 sampling iterations and the simulation study uses 750; the benchmark uses 500 to reduce runtime while providing representative timing comparisons across $C$ values.

\begin{table}[htbp]
\centering
\caption{Empirical sampling time as a function of the number of compositional categories $C$. $D = C - 1$ is the ILR dimension. Results are means over 3 independent runs ($T = 120$, 4 chains, 500 warmup + 500 sampling iterations). The relative column is normalized to $C = 5$.}
\label{tab:timing}
\begin{tabular}{ccccc}
\toprule
$C$ & $D = C - 1$ & Mean time (sec) & SD (sec) & Relative to $C = 5$ \\
\midrule
 5 &  4 & 36.3 &  2.7 & 1.00 \\
 7 &  6 & 36.5 &  4.5 & 1.01 \\
10 &  9 & 57.9 & 15.4 & 1.60 \\
15 & 14 & 65.1 &  1.2 & 1.79 \\
\bottomrule
\end{tabular}
\end{table}

Scaling is substantially more favorable than the $O(C^2)$ worst-case implied by the ILR matrix operations. The near-identical times at $C = 5$ and $C = 7$ suggest that MCMC overhead dominates at small $C$, with ILR and DARMA recursion costs becoming relevant only from $C = 10$ onward. The 1.8$\times$ increase from $C = 5$ to $C = 15$ indicates that the model scales comfortably to the 10-component lead-time application in this paper and remains tractable for moderate expansions beyond that. Applications with $C \gtrsim 20$ would warrant further benchmarking; in such cases, the diagonal AR/MA restriction provides meaningful computational relief relative to full-matrix VARMA specifications.

\section{Forecasting}
\label{sec:forecasting}

For each posterior draw, we compute gate values $w_{T+h}$, propagate DARMA dynamics forward with interventions included, and draw from the predictive Dirichlet \citep{west1997}. Aggregating across draws yields point forecasts, intervals, and densities. We evaluate forecasts using: (1) Aitchison distance for point accuracy; (2) mean absolute error (MAE) for componentwise accuracy; (3) energy score in Aitchison geometry for probabilistic accuracy \citep{gneiting2007}; (4) plug-in log score for density forecasts; and (5) empirical coverage of posterior predictive intervals for observed compositions $Y_t$. See Appendix~\ref{app:metrics} for formal definitions and \citet{hyndman2021} for general forecasting principles. The plug-in log score should therefore be interpreted as a computationally convenient approximation rather than the full posterior predictive log score, because it evaluates the Dirichlet density at posterior mean parameters and does not integrate over parameter uncertainty.

\section{Empirical Application: Fee Recognition Lead Times During COVID-19}
\label{sec:empirical}

We apply the directional-shift model to monthly distributions of fee recognition lead times (the delay between booking and revenue recognition) during the COVID-19 pandemic across four global regions.

\subsection{Data and Context}
\label{subsec:data_context}

Travel and hospitality businesses earn fees when services are rendered, not when bookings are made. The distribution of lead times (how far in advance customers book) is a compositional time series that drives revenue forecasting and financial planning \citep{song2008, athanasopoulos2011}.

Our data comprises monthly lead-time distributions for four geographic regions: North America (NAMER), Europe/Middle East/Africa (EMEA), Latin America (LATAM), and Asia-Pacific (APAC), from January 2014 through January 2021. Each observation is a 10-part composition representing the proportion of bookings with lead times falling into monthly buckets: 0, 1, 2, 3, 4, 5, 6, 7, 8, and 9+ months. The multi-region design allows us to assess whether the intervention model's advantages generalize across markets with potentially different booking dynamics and COVID impact patterns.

The COVID-19 pandemic induced a dramatic structural break beginning March 2020, shifting booking behavior sharply toward shorter lead times. This setting provides an ideal test case: the break date is known, the expected direction of change is clear (toward shorter lead times), and the transition was neither instantaneous nor complete within our sample period. We set $\ell$ to February 2020, so March 2020 ($t = \ell + 1$) is the first post-break observation.

Figure~\ref{fig:data_heatmap} displays the lead-time composition over time as a heatmap for the North American market. The structural break in March 2020 is visually apparent, with the transition clearly gradual rather than instantaneous.

\begin{figure}[htbp]
\centering
\includegraphics[width=0.95\textwidth]{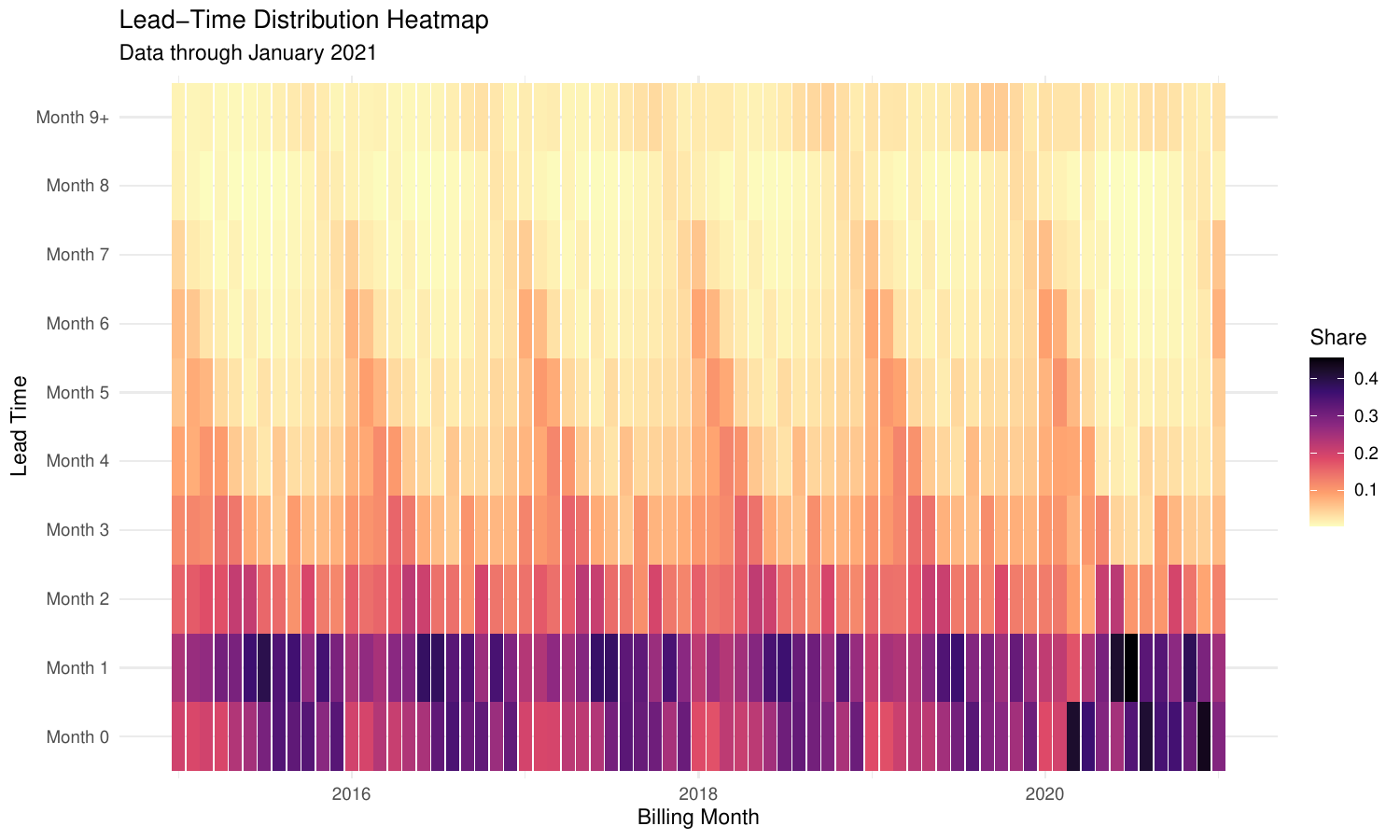}
\caption{Lead-time distribution over time (heatmap) for North America. Each column represents a month; each row represents a lead-time category (0 months at bottom, 9+ months at top). Darker shading indicates higher share. The COVID-19 structural break (March 2020) is clearly visible as a gradual shift toward shorter lead times. Similar patterns are observed across EMEA, LATAM, and APAC regions.}
\label{fig:data_heatmap}
\end{figure}

\subsection{Model Specifications}
\label{subsec:model_specs}

We compare three specifications:

\paragraph{Baseline model.} Standard B-DARMA with $P = Q = 1$ diagonal DARMA dynamics. The mean covariates $X_t$ include a linear time trend (normalized to $[0,1]$) and sine/cosine terms at 12-, 6-, and 4-month periods (three Fourier harmonics); the intercept is captured by $\bm{b}$. The precision model uses an intercept only ($X_{\phi,t} = 1$). No break mechanism.

\paragraph{Fixed effect model.} Baseline plus a post-COVID level shift. Specifically, we add a covariate $X_t^{\text{covid}} = \mathbf{1}(t > \ell)$ to the mean model with coefficient vector $\bm{\beta}_{\text{covid}} \in \R^{C-1}$, so the ILR-space mean becomes $\bm{\eta}_t = \bm{b} + \mathbf{B}X_t + \bm{\beta}_{\text{covid}} X_t^{\text{covid}}$. The precision model also includes a post-COVID shift: $\log \lambda_t = \gamma_0 + \delta_\phi^{\text{FE}} \cdot \mathbf{1}(t > \ell)$. The key distinction from our intervention model is that this approach assumes an \emph{instantaneous} step change at $t = \ell + 1$, whereas the intervention model captures \emph{gradual} transitions via the logistic gate.

\paragraph{Intervention model.} Baseline plus the directional-shift mechanism with logistic gate. Parameters $(\bm{v}, \Delta, \tau, \kappa, \delta_\phi)$ are estimated from the data.

All models use the same priors for shared parameters to ensure fair comparison. The mean covariate specification (linear trend plus three Fourier harmonics at annual, semi-annual, and quarterly periods) was selected to capture the dominant seasonal and trend variation in booking behavior. We deliberately excluded external economic indicators (travel demand indices, mobility data, policy stringency measures) from the covariate set: in the context of this paper, isolating the pure effect of the structural break is the goal of the intervention mechanism itself, and including concurrent economic covariates that are themselves COVID-affected would conflate the break signal with its drivers. Practitioners with access to leading indicators that are not themselves disrupted by the break may benefit from richer covariate specifications.

\subsection{Rolling Forecast Evaluation}
\label{subsec:rolling_eval}

To assess forecasting performance during the structural break, we conduct a rolling forecast evaluation across all four regions. For each forecast origin from July 2020 through January 2021 (7 origins per region, 28 total), we fit all three models using data up to but not including the origin month, generate $h$-step-ahead forecasts, and compare forecasts to observed values using our evaluation metrics (Appendix~\ref{app:metrics}).

This design tests each model's ability to forecast through an ongoing structural break, exactly the situation where the intervention model should excel.

\subsection{Results}
\label{subsec:empirical_results}

Table~\ref{tab:empirical_results_by_region} presents the rolling 1-step-ahead evaluation results separately for each region.

\begin{table}[htbp]
\centering
\caption{Rolling 1-step-ahead forecast evaluation by region, July 2020--January 2021 ($n=7$ origins per region). Lower is better for Aitchison distance, energy score, and MAE. Higher is better for plug-in log score. Coverage is componentwise, averaged across components; nominal is 80\%.}
\label{tab:empirical_results_by_region}
\begin{tabular}{llccccc}
\toprule
Region & Model & Aitchison & Energy & Plug-in & MAE & Coverage \\
 &  & Distance & Score & Log Score & & (80\%) \\
\midrule
\multirow{3}{*}{NAMER}
 & Baseline     & 1.212 & 0.801 & 18.2 & 0.0231 & 52.9\% \\
 & Fixed Effect & 0.915 & 0.621 & 21.9 & 0.0186 & 67.1\% \\
 & Intervention & \textbf{0.913} & \textbf{0.623} & \textbf{26.5} & \textbf{0.0167} & \textbf{78.6\%} \\
\addlinespace
\multirow{3}{*}{EMEA}
 & Baseline     & 1.203 & 0.788 & 18.5 & 0.0229 & 57.1\% \\
 & Fixed Effect & 0.934 & 0.645 & 21.8 & 0.0183 & 64.3\% \\
 & Intervention & \textbf{0.931} & \textbf{0.634} & \textbf{26.4} & \textbf{0.0174} & \textbf{78.6\%} \\
\addlinespace
\multirow{3}{*}{LATAM}
 & Baseline     & 1.212 & 0.797 & 18.1 & 0.0231 & 52.9\% \\
 & Fixed Effect & 0.996 & 0.628 & 18.2 & 0.0197 & 68.6\% \\
 & Intervention & \textbf{0.903} & \textbf{0.610} & \textbf{26.3} & \textbf{0.0174} & \textbf{80.0\%} \\
\addlinespace
\multirow{3}{*}{APAC}
 & Baseline     & 1.207 & 0.797 & 18.3 & 0.0229 & 51.4\% \\
 & Fixed Effect & \textbf{0.884} & \textbf{0.598} & 22.1 & 0.0184 & 64.3\% \\
 & Intervention & 0.937 & 0.625 & \textbf{26.3} & \textbf{0.0175} & \textbf{81.4\%} \\
\bottomrule
\end{tabular}
\end{table}

Table~\ref{tab:empirical_results_pooled} presents the pooled results across all four regions.

\begin{table}[htbp]
\centering
\caption{Pooled rolling 1-step-ahead forecast evaluation across all four regions, July 2020--January 2021 ($n=28$ origins total).}
\label{tab:empirical_results_pooled}
\begin{tabular}{lccccc}
\toprule
Model & Aitchison & Energy & Plug-in & MAE & Coverage \\
 & Distance & Score & Log Score & & (80\%) \\
\midrule
Baseline     & 1.209 & 0.796 & 18.3 & 0.0230 & 53.6\% \\
Fixed Effect & 0.932 & 0.623 & 21.0 & 0.0188 & 66.1\% \\
Intervention & \textbf{0.921} & \textbf{0.623} & \textbf{26.4} & \textbf{0.0173} & \textbf{79.6\%} \\
\bottomrule
\end{tabular}
\end{table}

\paragraph{Point forecast accuracy.} Both the intervention and fixed effect models achieve similar point accuracy, with mean Aitchison distances of 0.921 and 0.932 respectively, both representing substantial improvements over the baseline (1.209). The intervention model shows slightly better point accuracy in three of four regions (NAMER, EMEA, LATAM), while the fixed effect model performs better in APAC.

\paragraph{Probabilistic forecast quality.} The intervention model achieves a substantially higher plug-in log score (26.4 vs.\ 21.0 pooled), consistent in direction across all four regions. The energy scores are similar between models, as expected given comparable point accuracy.

\paragraph{Calibration.} The intervention model achieves near-nominal 80\% coverage across all four regions (78.6--81.4\%, pooled 79.6\%). In contrast, the fixed effect model shows substantial under-coverage (64.3--68.6\%, pooled 66.1\%), and the baseline model is severely miscalibrated (pooled 53.6\%). The calibration gap reflects the fixed effect model treating the break as instantaneous: on the first post-break observation, the full shift is applied, which misspecifies the gradual transition that actually occurred and produces intervals that are systematically too narrow during the adjustment period. The intervention model captures gradual transition dynamics via the logistic gate, propagating uncertainty about the transition trajectory into the forecast intervals.

\begin{figure}[htbp]
\centering
\includegraphics[width=0.85\textwidth]{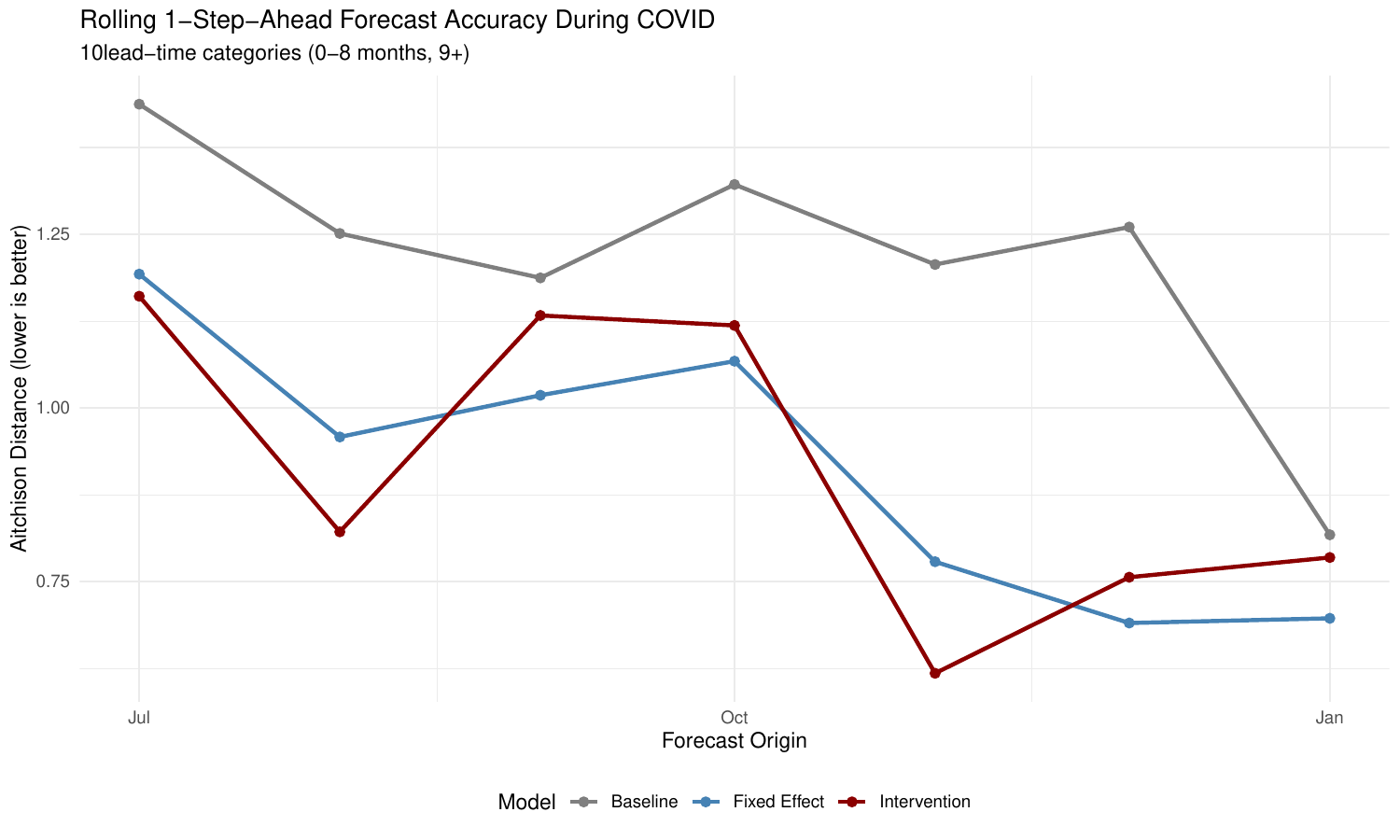}
\caption{Rolling 1-step-ahead Aitchison distance by forecast origin for North America. Both break-aware models show substantial improvement over baseline; the intervention model's primary advantage in uncertainty quantification emerges in coverage statistics (Table~\ref{tab:empirical_results_by_region}).}
\label{fig:rolling_comparison}
\end{figure}

\begin{figure}[htbp]
\centering
\includegraphics[width=0.95\textwidth]{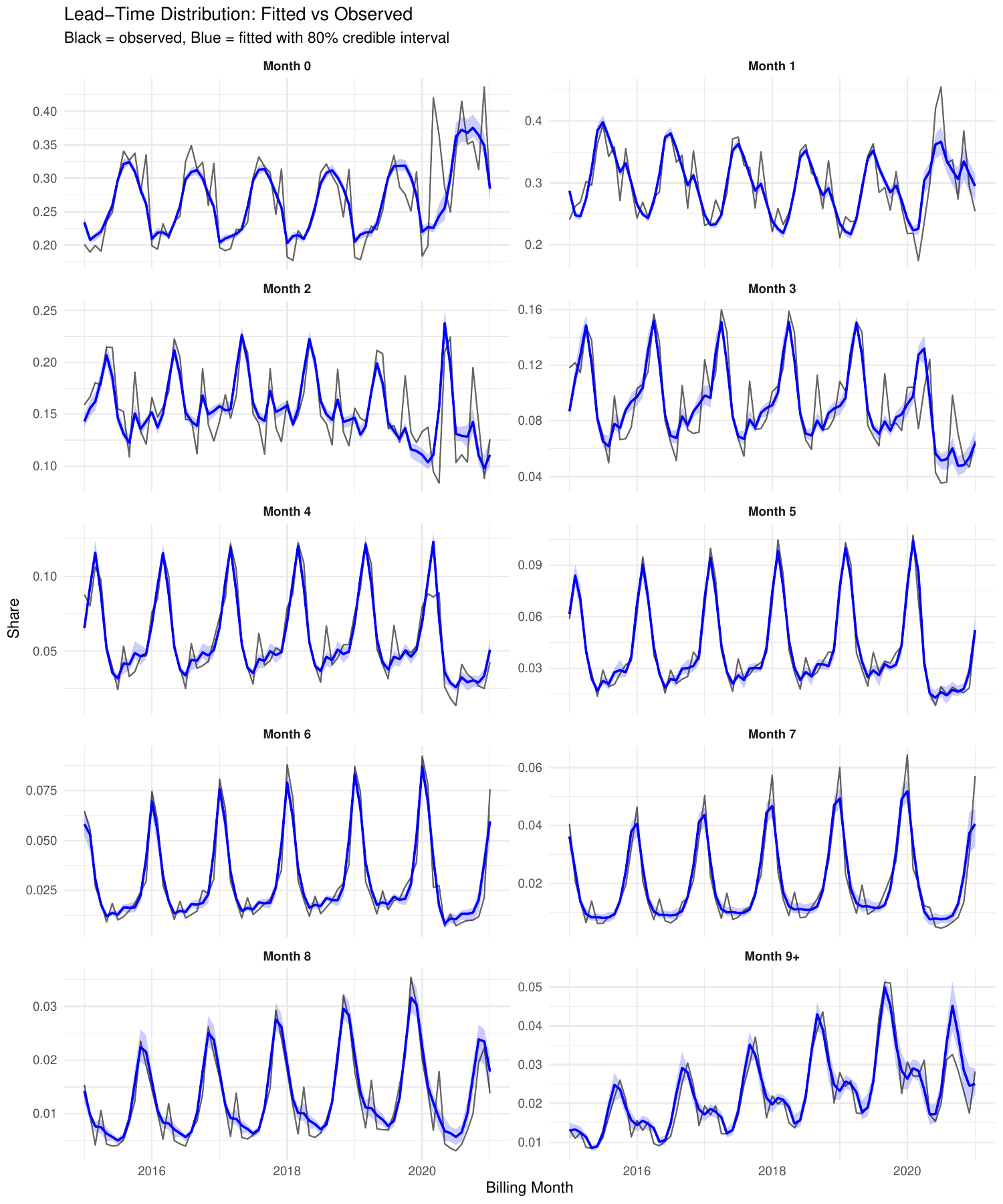}
\caption{In-sample model fit for North America: posterior mean of $\bm{\mu}_t$ (blue) versus observed lead-time shares (black) for all categories. Shaded bands show 80\% credible intervals. The intervention model captures both the pre-COVID seasonal patterns and the gradual structural shift beginning March 2020.}
\label{fig:fitted}
\end{figure}

\subsection{Residual Cross-Correlation Diagnostic}
\label{subsec:crosscorr}

The diagonal AR/MA restriction is the most consequential structural assumption in our implementation: it asserts that each ILR coordinate evolves independently given the drift term. Because components must sum to unity, cross-component dynamics are structurally embedded in the data, and the diagonal assumption may leave non-trivial dependence in the residuals. We assess this using fitted-value ILR residuals $\hat{r}^{\mathrm{fit}}_{t,d} = \ilr(Y_t)_d - \ilr(\hat{\bm{\mu}}_t)_d$ for the full in-sample fit of the intervention model in each region, where $\hat{\bm{\mu}}_t$ denotes the posterior mean fitted composition. These diagnostic residuals are not identical to the working residuals $e_t = Z_t - \eta_t$ used in the DARMA recursion: because $\E[Z_t \mid \mu_t, \lambda_t] \neq \eta_t$ in general, they should be interpreted as a fitted-value approximation in ILR space.

Table~\ref{tab:crosscorr} reports the results. Mean absolute pairwise cross-correlation ranges from 0.165 (APAC) to 0.186 (LATAM), with maximum values reaching 0.52. Ljung-Box portmanteau tests on the cross-product series are significant at $p < 0.05$ in 53--61\% of ILR dimension pairs across regions.

\begin{table}[htbp]
\centering
\caption{Fitted-value ILR residual cross-correlation diagnostic for the intervention model, lead-time application. Each region has $C = 10$ components and $D = 9$ ILR dimensions, yielding 36 pairs. Mean and maximum absolute Pearson correlation at lag 0; percent of pairs with significant Ljung-Box test ($p < 0.05$, lag 6).}
\label{tab:crosscorr}
\begin{tabular}{lcccc}
\toprule
Region & $n$ pairs & Mean $|r|$ & Max $|r|$ & \% Significant \\
\midrule
NAMER & 36 & 0.177 & 0.442 & 58.3\% \\
EMEA  & 36 & 0.185 & 0.495 & 61.1\% \\
LATAM & 36 & 0.186 & 0.522 & 55.6\% \\
APAC  & 36 & 0.165 & 0.484 & 52.8\% \\
\bottomrule
\end{tabular}
\end{table}

These results confirm that non-trivial cross-component residual dependence remains and that the diagonal assumption is not innocuous in this application. The observed residual cross-correlations can arise from both the structural simplex constraint, which induces dependence among components by construction, and genuine dynamic cross-component effects not captured by the diagonal parameterization. The practical consequence of that residual dependence is assessed directly in the full-matrix sensitivity benchmark of Section~\ref{subsec:fullvar_benchmark}.

\subsection{Full-Matrix Sensitivity Benchmark}
\label{subsec:fullvar_benchmark}

The residual diagnostic in Section~\ref{subsec:crosscorr} shows that the diagonal restriction is not innocuous, but it does not establish whether a denser dynamic specification improves forecasting. To assess the practical cost of the restriction directly, we fit a sensitivity benchmark in the lead-time application that retains the same Dirichlet likelihood, covariate specification, and rank-1 intervention $\Delta \cdot w_t \cdot \bm{v}$, but replaces the diagonal AR and MA operators with unrestricted $D \times D$ matrices in ILR space. To ensure a like-for-like comparison, we reran the diagonal intervention model under the same rolling-window pipeline used for the benchmark; the rerun reproduces the pooled 1-step results in Table~\ref{tab:empirical_results_pooled}.

\begin{table}[htbp]
\centering
\caption{Sensitivity benchmark for the lead-time application: pooled forecast evaluation for the diagonal intervention model and a dense full-matrix intervention benchmark. Both models use the same Dirichlet likelihood, covariates, and rank-1 directional-shift intervention; the only difference is the AR/MA operator structure. Diagonal results reproduce Tables~\ref{tab:empirical_results_pooled} and~\ref{tab:horizon_results}. Lower is better for Aitchison distance and MAE; coverage nominal is 80\%.}
\label{tab:fullvar_benchmark}
\begin{tabular}{llccc}
\toprule
Horizon & Model & Aitchison Dist & MAE & Coverage \\
\midrule
$h=1$  & Diagonal    & \textbf{0.921} & \textbf{0.0173} & \textbf{79.6\%} \\
(n=28) & Full matrix & 1.461          & 0.0289          & 44.6\%          \\
\addlinespace
$h=3$  & Diagonal    & \textbf{0.947} & \textbf{0.0184} & \textbf{75.0\%} \\
(n=20) & Full matrix & 1.355          & 0.0295          & 46.0\%          \\
\addlinespace
$h=6$  & Diagonal    & \textbf{0.768} & \textbf{0.0204} & \textbf{75.0\%} \\
(n=8)  & Full matrix & 1.174          & 0.0286          & 56.3\%          \\
\bottomrule
\end{tabular}
\end{table}

The dense full-matrix benchmark performs materially worse at every forecast horizon.
Pooled 1-step-ahead Aitchison distance deteriorates from 0.921 to 1.461, MAE from
0.0173 to 0.0289, and coverage from 79.6\% to 44.6\%. The same ranking holds at
$h=3$ and $h=6$, with the full-matrix coverage remaining below 60\% throughout while
the diagonal model stays within a few percentage points of nominal. These results do
not imply that cross-component dynamics are absent. Rather, they show that naively
estimating unrestricted AR/MA operators in these rolling windows produces a poor
bias-variance tradeoff. With $D=9$ ILR dimensions and $P=Q=1$, the full-matrix
specification adds 144 free parameters relative to the diagonal, far more than the
post-break sample can support. In the present application, the diagonal specification
functions as a useful regularizing approximation for forecasting.

\subsection{Forecast Horizon Analysis}
\label{subsec:horizon_analysis}

Table~\ref{tab:horizon_results} summarizes accuracy by horizon, pooled across all four regions.

\begin{table}[htbp]
\centering
\caption{Forecast accuracy by horizon (months ahead), pooled across all four regions. Sample sizes: $h=1$, $n=28$; $h=3$, $n=20$; $h=6$, $n=8$.}
\label{tab:horizon_results}
\begin{tabular}{llccc}
\toprule
Horizon & Model & Aitchison Dist & MAE & Coverage \\
\midrule
$h=1$ & Baseline     & 1.209 & 0.0230 & 53.6\% \\
(n=28) & Fixed Effect & 0.932 & 0.0188 & 66.1\% \\
      & Intervention & \textbf{0.921} & \textbf{0.0173} & \textbf{79.6\%} \\
\addlinespace
$h=3$ & Baseline     & 1.196 & 0.0245 & 39.0\% \\
(n=20) & Fixed Effect & \textbf{0.910} & 0.0205 & 57.5\% \\
      & Intervention & 0.947 & \textbf{0.0184} & \textbf{75.0\%} \\
\addlinespace
$h=6$ & Baseline     & 1.067 & 0.0296 & 43.8\% \\
(n=8) & Fixed Effect & \textbf{0.664} & \textbf{0.0148} & 56.3\% \\
      & Intervention & 0.768 & 0.0204 & \textbf{75.0\%} \\
\bottomrule
\end{tabular}
\end{table}

At $h=1$, the intervention model achieves the best point accuracy. At longer horizons ($h=3$, $h=6$), the fixed effect model shows better Aitchison distance as forecasts target dates well after the transition has largely completed. The intervention model maintains its calibration advantage across all horizons: 79.6\% at $h=1$, 75.0\% at $h=3$, and 75.0\% at $h=6$, versus 66.1\%, 57.5\%, and 56.3\% for the fixed effect. The calibration gap widens slightly at longer horizons: 13.5 percentage points at $h=1$, 17.5 at $h=3$, and 18.7 at $h=6$, as the fixed effect model continues to under-cover throughout the transition period.

\begin{figure}[htbp]
\centering
\includegraphics[width=0.85\textwidth]{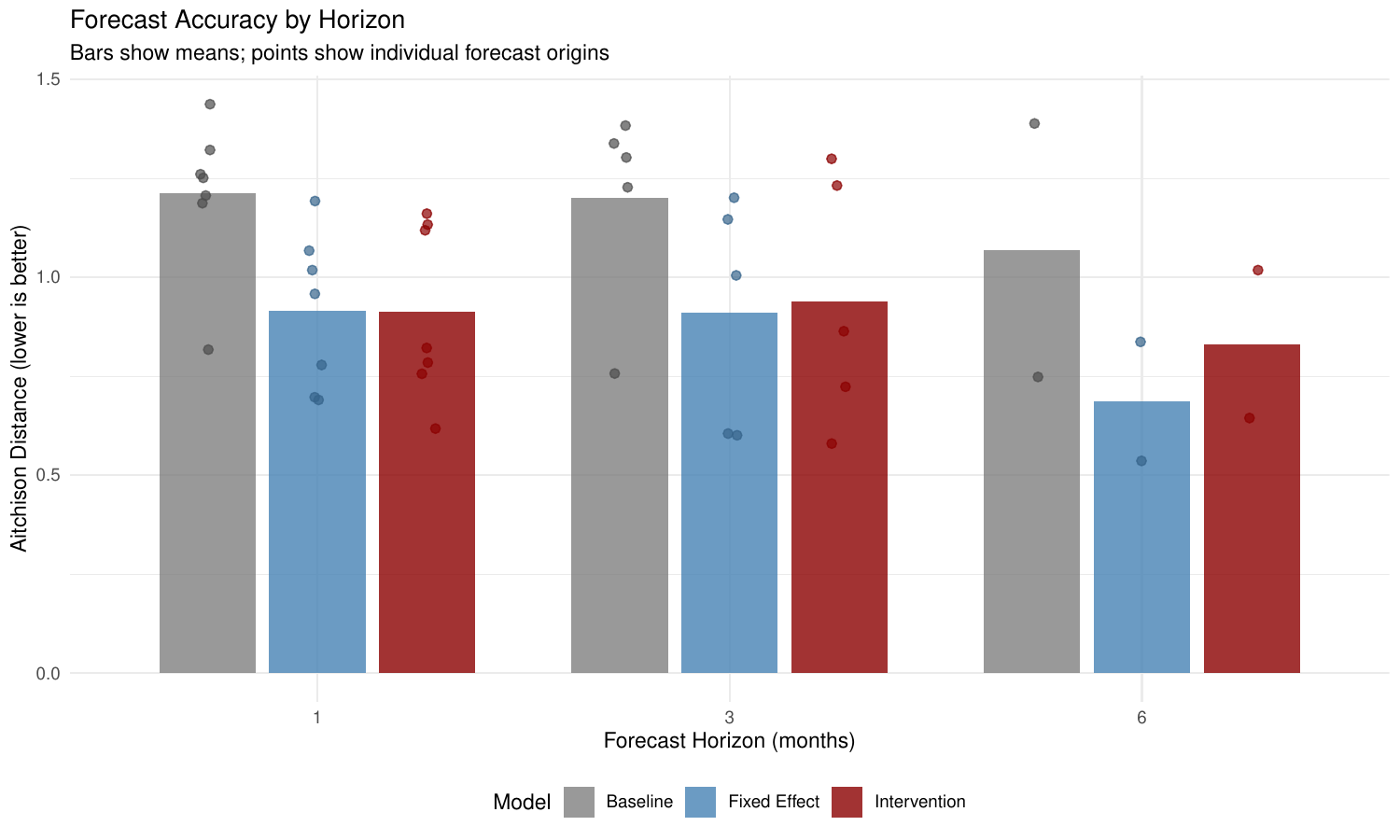}
\caption{Mean Aitchison distance by forecast horizon, pooled across all four regions. At longer horizons, the fixed effect model achieves better point accuracy. Calibration results by horizon are reported in Table~\ref{tab:horizon_results}.}
\label{fig:horizon}
\end{figure}

\section{Second Empirical Application: Stay-Length Composition During COVID-19}
\label{sec:empirical2}

The first empirical application provided a favorable test case: the break date was known, the direction of change was unambiguous ex ante, and the post-break dynamics were largely monotone. The second application examines accommodation stay-length distributions, where both conditions are relaxed. The direction of compositional shift is less obvious ex ante, and the post-break dynamics are demonstrably non-monotone. Our goal is to characterize what the model wins and loses under these harder conditions, not to demonstrate a second unqualified success.

\subsection{Data and Motivation}
\label{subsec:los_data}

The distribution of stay lengths (how long guests book accommodations) is a compositional time series that shifted substantially during COVID-19 \citep{katz2025slomads}. Remote work, travel uncertainty, and changing preferences induced a complex redistribution across stay-length categories that evolved in multiple stages rather than a single monotone transition.

Our data comprise monthly stay-length distributions for three geographic regions (NAMER, EMEA, and LATAM) from January 2018 through December 2022. Each observation is a 7-part composition representing the share of bookings falling into the buckets: 1 night, 2 nights, 3 nights, 4--7 nights, 8--14 nights, 15--27 nights, and 28+ nights. We set $\ell$ to February 2020.

This application presents harder conditions along two dimensions. First, the direction of compositional change is ambiguous ex ante: the 28+ night share was small ($\approx 1.5\%$ pre-COVID) and the redistribution across seven buckets is not obvious without examining the data. Second, the post-break dynamics are non-monotone: the 28+ night category spiked sharply in mid-2020, then partially reverted through 2021--2022 as travel normalized. Figure~\ref{fig:los_recovery} illustrates this arc, with LATAM and NAMER showing the most pronounced spike-and-reversion pattern. This reversion lies outside the scope of the monotone logistic gate. The key question is not whether the model succeeds but precisely what it gains and loses relative to the simpler fixed-effect alternative, a question our reversibility simulation (Section~\ref{subsec:sim_reversibility}) bounds from below.

We exclude APAC from the main analysis. Exploratory analysis (Appendix~\ref{app:los_apac}) shows that APAC's stay-length composition has minimal COVID-19 signal, with the 28+ bucket essentially flat throughout 2020--2022; including it would conflate compositional stability with model performance.

\begin{figure}[htbp]
\centering
\includegraphics[width=0.95\textwidth]{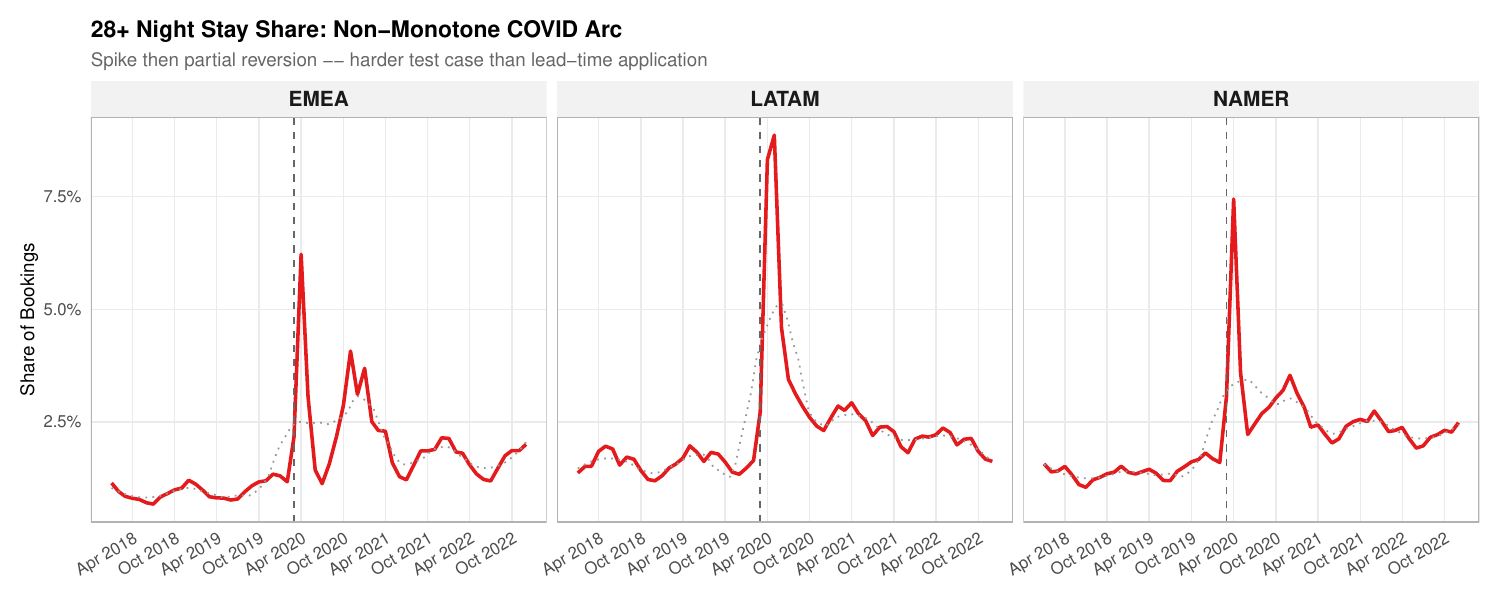}
\caption{Share of bookings with stay lengths of 28+ nights for NAMER, EMEA, and LATAM, January 2018--December 2022. Dashed line = COVID-19 structural break (March 2020). All three regions exhibit a spike followed by partial reversion, a non-monotone arc that lies outside the scope of the monotone logistic gate and motivates a precise comparison of what the intervention model gains and loses relative to the fixed-effect alternative.}
\label{fig:los_recovery}
\end{figure}

Figure~\ref{fig:los_heatmap} displays the full 7-component stay-length heatmaps for all four regions.

\begin{figure}[htbp]
\centering
\includegraphics[width=0.95\textwidth]{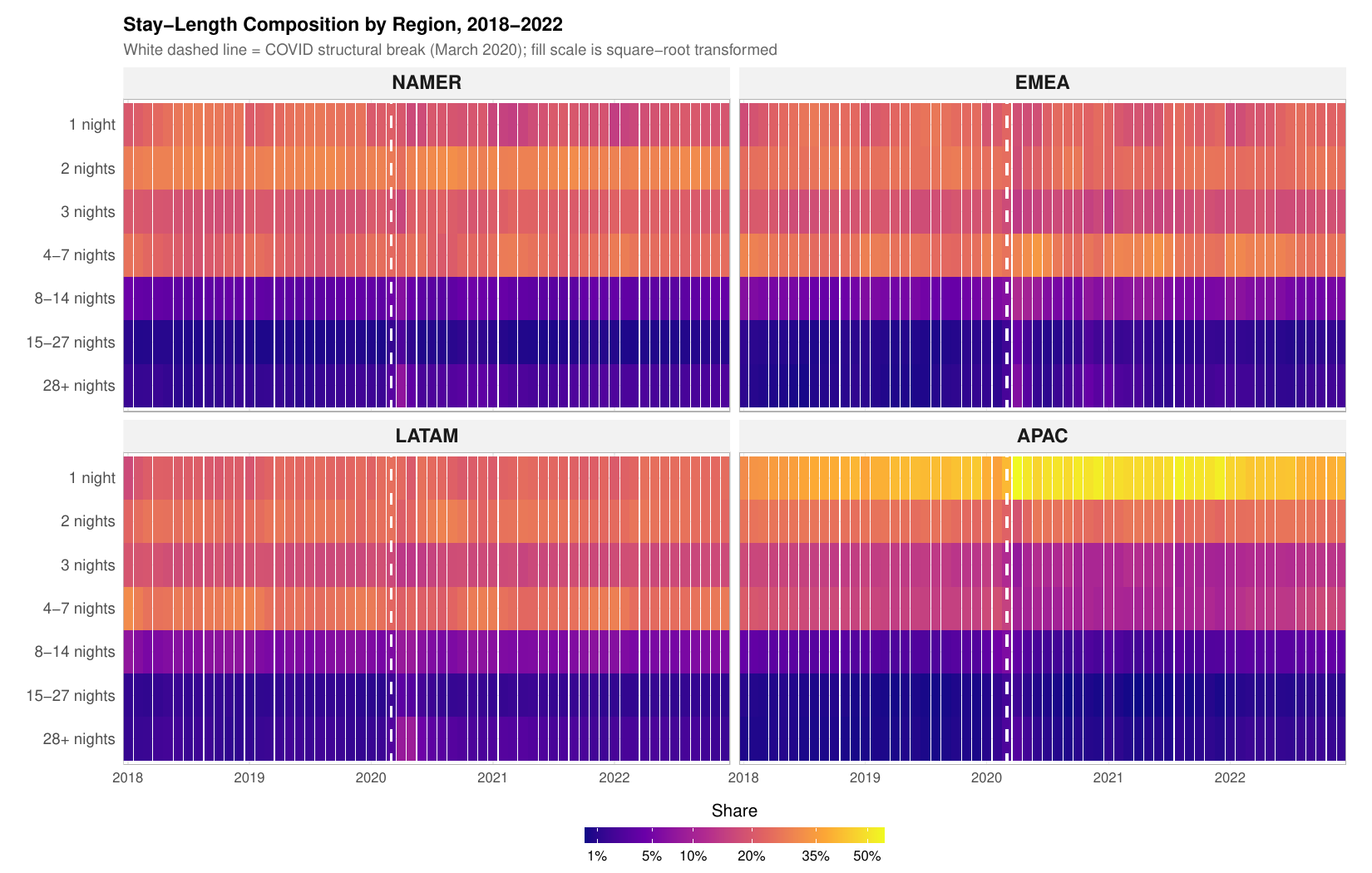}
\caption{Stay-length composition heatmaps for all four regions, January 2018--December 2022. White dashed line = COVID-19 structural break (March 2020). The signal is clear in NAMER, EMEA, and LATAM; APAC shows minimal compositional change and is excluded from the main forecast evaluation.}
\label{fig:los_heatmap}
\end{figure}

\subsection{Model Specifications and Rolling Evaluation}
\label{subsec:los_models}

We compare the same three specifications as in Section~\ref{sec:empirical}: baseline B-DARMA($1,1$) with no break mechanism, a fixed effect model with a post-COVID step dummy appended to the mean covariate vector, and the directional-shift intervention model with logistic gate. All models use Fourier seasonal covariates (three harmonics plus intercept) for the mean and a trend-seasonal model for the precision equation.

We conduct a rolling forecast evaluation with 7 origins per region (July 2020 through January 2021, 21 origins total) at horizons $h \in \{1, 2, 3, 4, 5, 6\}$ months. The evaluation window spans the active transition period, during which the non-monotone dynamics are most pronounced.

\subsection{Residual Cross-Correlation Diagnostic}
\label{subsec:los_crosscorr}

We apply the same fitted-value ILR residual diagnostic as in Section~\ref{subsec:crosscorr} to the stay-length intervention model fits. Table~\ref{tab:crosscorr_los} reports pairwise cross-correlations of $\hat{r}^{\mathrm{fit}}_{t,d} = \ilr(Y_t)_d - \ilr(\hat{\bm{\mu}}_t)_d$ for the full training window at origin 7 (the most data-rich fit, $T \approx 36$ months, $C = 7$).

\begin{table}[htbp]
\centering
\caption{Fitted-value ILR residual cross-correlation diagnostic for the intervention model, stay-length application. Each region has $C = 7$ components and $D = 6$ ILR dimensions, yielding 15 pairs. Training window at origin 7 is approximately $T = 36$ months.}
\label{tab:crosscorr_los}
\begin{tabular}{lcccc}
\toprule
Region & $n$ pairs & Mean $|r|$ & Max $|r|$ & \% Significant \\
\midrule
NAMER & 15 & 0.417 & 0.926 & 20.0\% \\
EMEA  & 15 & 0.459 & 0.871 &  0.0\% \\
LATAM & 15 & 0.460 & 0.889 &  0.0\% \\
\bottomrule
\end{tabular}
\end{table}

The raw cross-correlations are substantially higher than in the lead-time application (mean 0.445 vs.\ 0.178 pooled), but Ljung-Box tests are mostly insignificant (6.7\% of pairs overall vs.\ 56.9\% in the lead-time application). The apparent paradox is explained by two factors. First, with $C = 7$ buckets heavily concentrated in short stays (1--7 nights account for 85--95\% of bookings), the ILR coordinates are highly correlated by the simplex constraint alone, independent of any dynamic misspecification. Second, the training windows are short ($T \approx 36$ months) at the evaluation origins, giving the portmanteau test low power to detect dynamic dependence even if present. Together, these results suggest that the high raw correlations in the stay-length application reflect structural constraint properties of the composition rather than omitted cross-component dynamics, and the low significance rates are consistent with the diagonal assumption being more innocuous here than in the longer lead-time series.

\subsection{Results}
\label{subsec:los_results}

Tables~\ref{tab:los_by_region} and~\ref{tab:los_pooled} present the 1-step-ahead results by region and pooled.

\begin{table}[htbp]
\centering
\caption{Rolling 1-step-ahead forecast evaluation by region, July 2020--January 2021 ($n=7$ origins per region). Lower is better for Aitchison distance, energy score, and MAE; higher is better for plug-in log score. Nominal coverage is 80\%. APAC excluded; see Appendix~\ref{app:los_apac}.}
\label{tab:los_by_region}
\begin{tabular}{llccccc}
\toprule
Region & Model & Aitchison & Energy & Plug-in & MAE & Coverage \\
 & & Distance & Score & Log Score & & (80\%) \\
\midrule
\multirow{3}{*}{EMEA}
 & Baseline     & 0.434 & 0.379 & $-$9.7  & 0.0128 & 53.1\% \\
 & Fixed Effect & 0.327 & 0.292 & 1.4 & \textbf{0.0091} & 81.6\% \\
 & Intervention & \textbf{0.320} & \textbf{0.235} & \textbf{10.3} & 0.0120 & \textbf{79.6\%} \\
\addlinespace
\multirow{3}{*}{LATAM}
 & Baseline     & 0.265 & 0.204 & 2.6  & 0.0098 & 65.3\% \\
 & Fixed Effect & \textbf{0.149} & \textbf{0.132} & \textbf{16.4} & \textbf{0.0061} & \textbf{79.6\%} \\
 & Intervention & 0.268 & 0.188 & 9.7  & 0.0075 & 83.7\% \\
\addlinespace
\multirow{3}{*}{NAMER}
 & Baseline     & 0.324 & 0.264 & $-$8.2 & 0.0108 & \textbf{81.6\%} \\
 & Fixed Effect & 0.159 & \textbf{0.135} & 20.9 & \textbf{0.0053} & 87.8\% \\
 & Intervention & \textbf{0.154} & 0.138 & \textbf{21.5} & 0.0056 & 93.9\% \\
\bottomrule
\end{tabular}
\end{table}

\begin{table}[htbp]
\centering
\caption{Pooled rolling 1-step-ahead forecast evaluation across EMEA, LATAM, and NAMER ($n=21$ origins total).}
\label{tab:los_pooled}
\begin{tabular}{lccccc}
\toprule
Model & Aitchison & Energy & Plug-in & MAE & Coverage \\
 & Distance & Score & Log Score & & (80\%) \\
\midrule
Baseline     & 0.341 & 0.283 & $-$5.1 & 0.0111 & 66.7\% \\
Fixed Effect & \textbf{0.212} & \textbf{0.186} & 12.9 & \textbf{0.0069} & \textbf{83.0\%} \\
Intervention & 0.247 & 0.187 & \textbf{13.8} & 0.0084 & 85.7\% \\
\bottomrule
\end{tabular}
\end{table}

\paragraph{Calibration.} Both break-aware models achieve near-nominal calibration in this application, in contrast to the baseline which substantially under-covers (66.7\%). At the pooled level, the fixed effect at 83.0\% is slightly closer to nominal than the intervention at 85.7\%. At the regional level: in EMEA the intervention is closest to nominal (79.6\% vs.\ 81.6\% for the fixed effect); in LATAM the fixed effect is closest (79.6\% vs.\ 83.7\% for the intervention); in NAMER both break-aware models over-cover substantially and the baseline at 81.6\% is ironically nearest to nominal. The key result is that both break-aware models are acceptably calibrated throughout, the fixed effect is marginally closer to nominal at the pooled level, and the baseline severely under-covers (66.7\% pooled).

\paragraph{Point accuracy.} The fixed effect model outperforms the intervention model on point accuracy in the pooled results (Aitchison distance 0.212 vs.\ 0.247) and in one of three regions (LATAM, where the fixed effect achieves 0.149 vs.\ 0.268 for the intervention). In EMEA the intervention is slightly better (0.320 vs.\ 0.327), and in NAMER the two models are essentially tied (0.154 vs.\ 0.159). The pooled disadvantage is driven primarily by LATAM, where the spike-and-reversion is both the sharpest and the most complete, leaving the fixed effect's flat post-break level as a better approximation of the 2021 equilibrium.

\paragraph{Log score.} The intervention model achieves a higher plug-in log score in two of three regions (EMEA and NAMER) and in the pooled results (13.8 vs.\ 12.9). LATAM is the exception, where the fixed effect's higher point accuracy translates to a higher log score (16.4 vs.\ 9.7). The pooled advantage reflects the intervention model's narrower intervals in EMEA, where it is closer to nominal, together with its log-score lead in NAMER despite over-covering there.

\begin{figure}[htbp]
\centering
\includegraphics[width=0.95\textwidth]{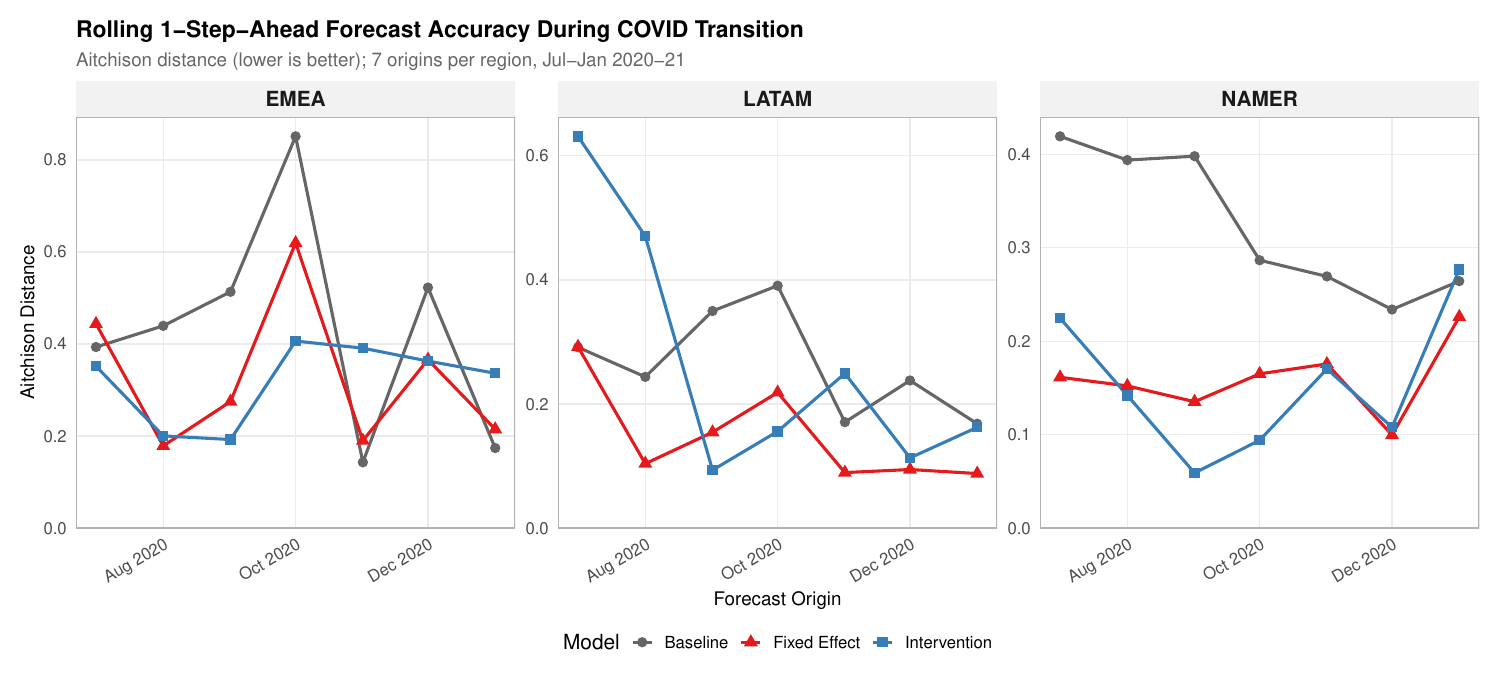}
\caption{Rolling 1-step-ahead Aitchison distance by forecast origin for EMEA, LATAM, and NAMER. The baseline is generally worse on average, especially in EMEA and NAMER. The fixed effect (red) achieves lower point error in LATAM; the intervention (blue) is slightly better in EMEA and essentially tied in NAMER. Calibration results by region are reported in Table~\ref{tab:los_by_region}.}
\label{fig:los_rolling}
\end{figure}

\begin{figure}[htbp]
\centering
\includegraphics[width=0.95\textwidth]{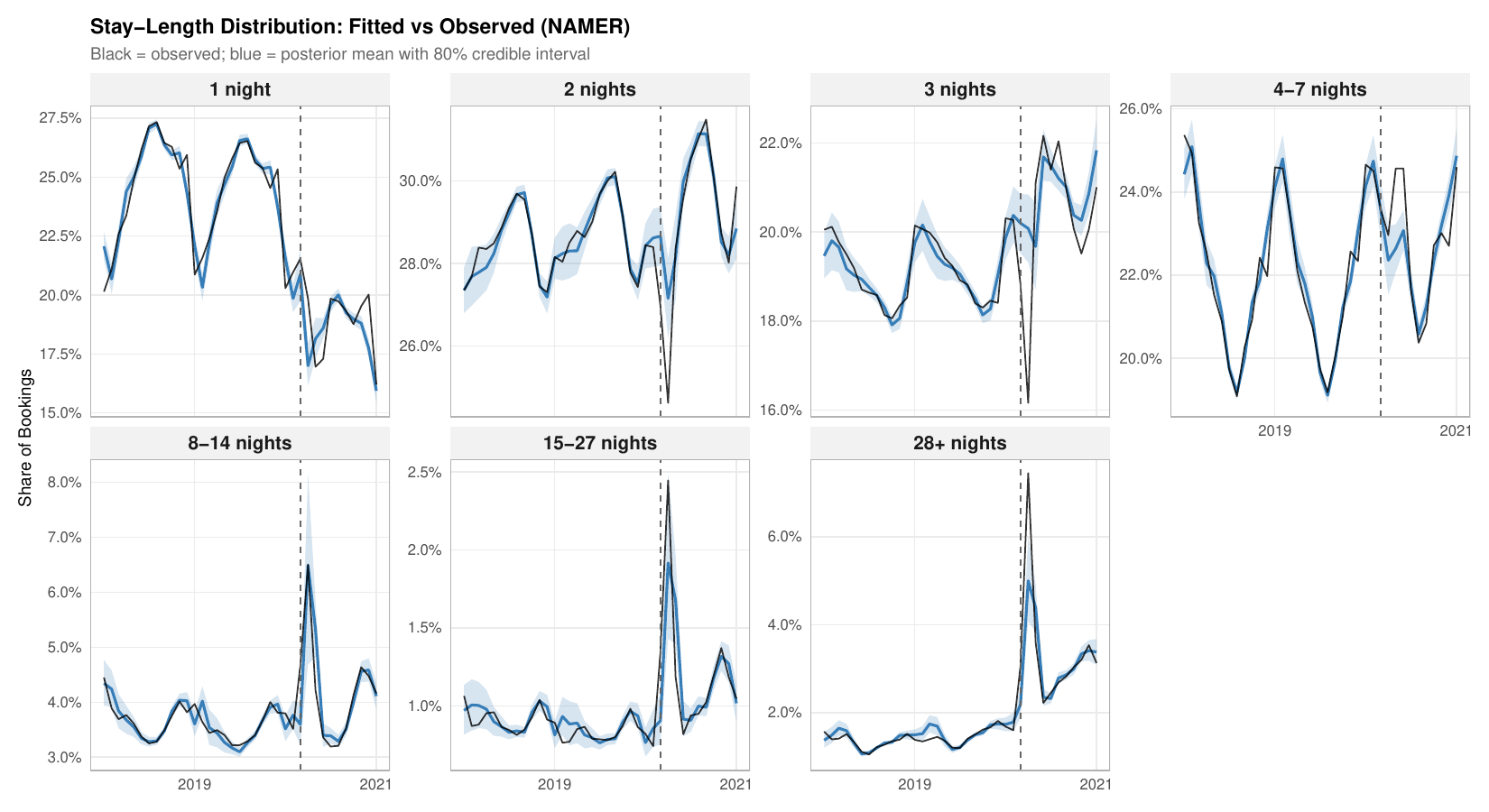}
\caption{In-sample model fit for NAMER: posterior mean (blue) versus observed stay-length shares (black) for all seven categories. Shaded bands show 80\% credible intervals. The intervention model captures both the pre-COVID seasonal patterns and the structural shift beginning March 2020 (dashed line). The partial recovery in the 28+ night category is visible in the lower-right panel; the model approximates the spike but cannot fully represent the reversion given the monotone gate constraint.}
\label{fig:los_fitted}
\end{figure}

\subsection{Horizon Analysis}
\label{subsec:los_horizon}

\begin{table}[htbp]
\centering
\caption{Forecast accuracy by horizon (months ahead), pooled across EMEA, LATAM, and NAMER ($n=21$ origins per horizon).}
\label{tab:los_horizon}
\begin{tabular}{llccc}
\toprule
Horizon & Model & Aitchison Dist & MAE & Coverage \\
\midrule
$h=1$ & Baseline     & 0.341 & 0.0111 & 66.7\% \\
      & Fixed Effect & \textbf{0.212} & \textbf{0.0069} & \textbf{83.0\%} \\
      & Intervention & 0.247 & 0.0084 & 85.7\% \\
\addlinespace
$h=3$ & Baseline     & 0.462 & 0.0132 & 71.4\% \\
      & Fixed Effect & \textbf{0.336} & \textbf{0.0096} & \textbf{81.0\%} \\
      & Intervention & 0.342 & 0.0103 & \textbf{81.0\%} \\
\addlinespace
$h=6$ & Baseline     & 0.437 & 0.0147 & 81.0\% \\
      & Fixed Effect & \textbf{0.385} & 0.0112 & \textbf{80.3\%} \\
      & Intervention & 0.448 & \textbf{0.0111} & 78.2\% \\
\bottomrule
\end{tabular}
\end{table}

At $h=1$ the fixed effect is slightly closer to nominal (83.0\% vs.\ 85.7\%); at $h=3$ both models are tied (81.0\% each); and at $h=6$ the fixed effect remains marginally closer (80.3\% vs.\ 78.2\%). The intervention's acceptable calibration weakens slightly at longer horizons, consistent with the reversibility simulation (Section~\ref{subsec:sim_reversibility}): as forecasts extend into the reversion period, the monotone gate increasingly misspecifies the dynamics.

\begin{figure}[htbp]
\centering
\includegraphics[width=0.95\textwidth]{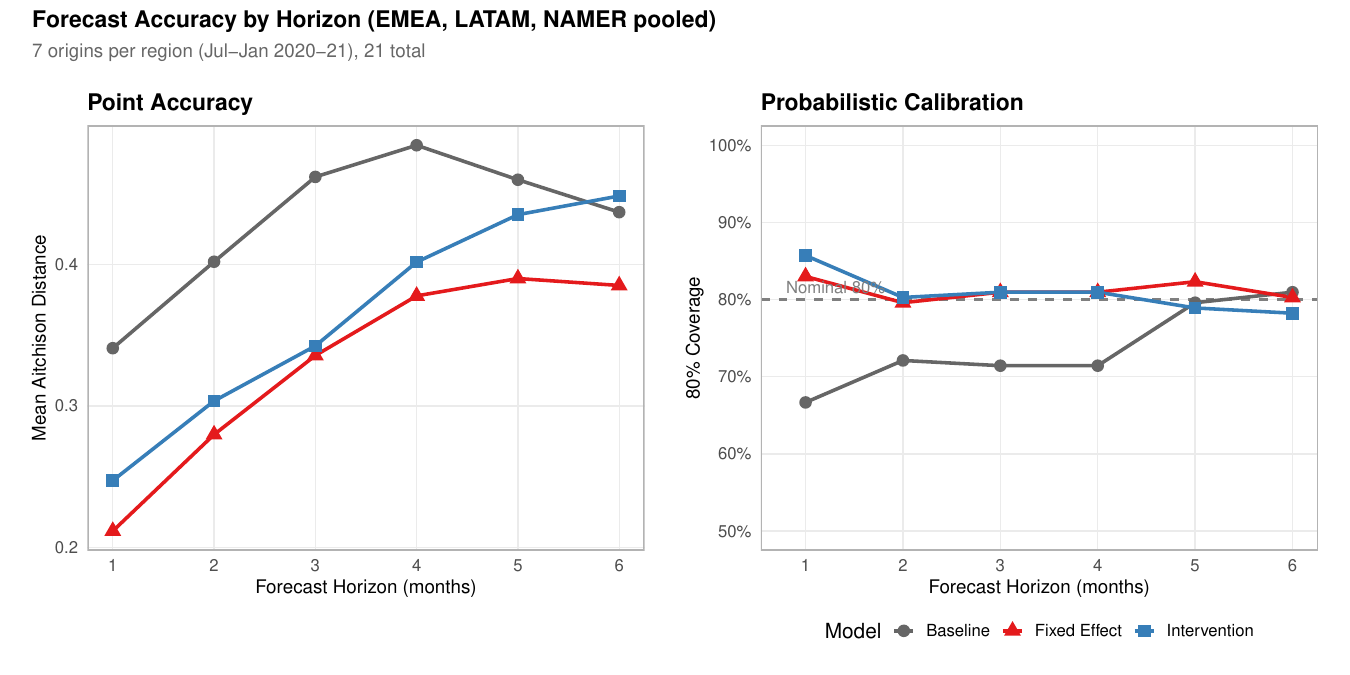}
\caption{Forecast accuracy by horizon for EMEA, LATAM, and NAMER pooled. Left: Aitchison distance (lower is better). Right: 80\% coverage (nominal indicated by dashed line). The fixed effect model achieves consistently lower point error. On coverage, the fixed effect is slightly closer to the 80\% nominal at $h=1$ (83.0\% vs.\ 85.7\%); the two models are tied at $h=3$ (81.0\% each); and the fixed effect remains marginally closer to nominal at $h=6$.}
\label{fig:los_horizon}
\end{figure}

\subsection{Interpretation and Comparison with Lead-Time Application}
\label{subsec:los_comparison}

The two empirical applications together yield a precise characterization of what the directional-shift model gains and loses relative to the fixed-effect alternative. The calibration story differs between the two applications. In the lead-time application, the intervention model is the only well-calibrated specification (79.6\% coverage vs.\ 66.1\% for the fixed effect), so the calibration advantage is clear and large. In the stay-length application, both break-aware models achieve acceptable calibration near nominal (83.0\% and 85.7\%), with the fixed effect actually slightly closer to the 80\% target. The intervention model's tendency to over-cover in the non-monotone setting reflects additional posterior uncertainty about the transition trajectory that is not fully resolved by the data. The point accuracy advantage is conditional: the intervention model matches or beats the fixed effect when the break is monotone and ongoing, but concedes point accuracy when the post-break dynamics include a reversion the monotone gate cannot represent.

This distinction has a clean structural explanation. The fixed effect model's step-function assumption is wrong during an ongoing transition but becomes approximately correct once the composition has settled into a post-break equilibrium. The logistic gate correctly delays full adjustment, which helps calibration throughout but only helps point accuracy while the gate is still rising. When the gate levels off and reality begins reverting, the fixed effect's flat level prediction becomes a better point approximation while the intervention model's wider intervals preserve calibration at some cost to point accuracy.

The practical decision rule follows directly: prefer the directional-shift model when calibrated uncertainty quantification is the primary objective and the break is expected to be roughly monotone. Prefer the fixed effect when the primary objective is point accuracy or the post-break dynamics may be non-monotone. Across the two applications, the main practical choice is usually between the two break-aware specifications rather than between break-aware and na\"{i}ve approaches, because the break-aware models improve materially on pooled point accuracy and on most forecast metrics, though not uniformly on every individual metric.

\section{Discussion}
\label{sec:discussion}

We developed a directional-shift extension to Bayesian Dirichlet ARMA that captures structural breaks through interpretable parameters: direction, amplitude, and timing. Simulations show accurate parameter recovery (conditional on direction identification) with nominal 80\% calibration across a wide range of transition speeds; calibration degrades gracefully under partial reversibility. The two empirical applications establish that the intervention model remains acceptably calibrated in both settings, that its calibration advantage over the fixed effect is specific to the monotone lead-time application, and that the point accuracy advantage is conditional on the transition being roughly monotone.

\paragraph{Point accuracy vs.\ calibration.} Our results highlight an important distinction between point forecast accuracy and probabilistic calibration. In the lead-time application the two break-aware models achieve similar point accuracy at $h=1$, and the fixed effect model achieves better point accuracy at longer horizons. In the stay-length application the fixed effect achieves better pooled point accuracy, though the advantage is concentrated in LATAM; the intervention matches or leads in EMEA and NAMER. In the lead-time application the intervention substantially outperforms the fixed effect on calibration (79.6\% vs.\ 66.1\%); in the stay-length application both models are near-nominal with the fixed effect marginally closer to the 80\% target. For applications where risk quantification matters (such as financial planning, resource allocation, or scenario analysis), the calibration advantage may be more important than modest differences in point accuracy.

\paragraph{Regional consistency.} In the lead-time application, the intervention model's calibration advantage is consistent across all four regions (78.6--81.4\%), suggesting the benefit reflects the smooth transition mechanism rather than any particular market structure.

\paragraph{Practical guidance.} The two empirical applications yield a precise characterization. The intervention model remains acceptably calibrated whether the break is monotone (lead-time) or non-monotone (stay-length), but its relative calibration advantage over the fixed effect appears only in the monotone lead-time application, where the fixed effect severely under-covers. The point accuracy advantage is conditional: it holds when the break is monotone and the transition is ongoing. In the stay-length application the fixed effect leads on pooled point accuracy, driven primarily by LATAM where the reversion is most complete, while the intervention matches or leads in EMEA and NAMER. Prefer the directional-shift model when calibrated uncertainty quantification is the primary objective and the break is expected to be roughly monotone. Prefer the fixed effect when the primary objective is point accuracy or the post-break dynamics may be non-monotone. In both applications, the break-aware models improve materially on pooled point accuracy and on most forecast metrics relative to the baseline, though not uniformly on every individual metric.

\paragraph{Choice of ILR basis.} Because the diagonal dynamics assume that each ILR coordinate evolves independently given the drift, the basis is a modeling choice rather than a neutral re-expression, and a sensible default is to align it with contrasts that are both interpretable and plausibly close to dynamically separate. When the categories carry a natural order, as with the duration-ordered lead-time and stay-length buckets here, the Helmert contrasts we use group adjacent categories and tend to keep cross-coordinate dependence modest. When no ordering is natural, a sequential binary partition built from subject-matter structure, or a data-driven basis such as principal balances aligned with the dominant directions of the residual cross-covariance \citep{egozcue2003, pawlowsky2015}, will reduce the off-diagonal dependence that the diagonal specification ignores. In practice we recommend choosing the basis before estimation from domain knowledge of which share movements are mechanically linked, then checking the fitted-value residual cross-correlations (Section~\ref{subsec:crosscorr}) to confirm the chosen coordinates are approximately uncoupled.

\paragraph{Limitations.} The model assumes a single directional shift with known break date; approximately 22.5\% of simulations fail to recover direction correctly (though calibration remains nominal). Two theoretical concessions deserve explicit acknowledgment.

First, our implementation restricts DARMA dynamics to diagonal AR and MA matrices, so each ILR coordinate evolves independently given the drift term. This is a strong assumption in compositional settings where components must sum to unity. The residual cross-correlation diagnostic confirms that the cost varies by application. In the lead-time application (Sections~\ref{subsec:crosscorr}--\ref{subsec:fullvar_benchmark}), mean absolute pairwise fitted-value ILR residual cross-correlation ranges from 0.165 to 0.186 across regions, with 53--61\% of pairs statistically significant, indicating that non-trivial cross-component dependence remains after fitting the diagonal model. In the stay-length application (Section~\ref{subsec:los_crosscorr}), raw correlations are higher (mean 0.445) but almost entirely insignificant, reflecting the structural simplex constraint of a composition dominated by short stays rather than clear dynamic misspecification. We also fit a dense full-matrix sensitivity benchmark in the lead-time application, holding the Dirichlet likelihood, covariates, and intervention structure fixed while replacing the diagonal AR/MA operators with unrestricted matrices. Despite the residual dependence, the full-matrix benchmark materially worsens pooled point accuracy and calibration at each reported forecast horizon. We therefore interpret the diagonal specification as a deliberate bias-variance tradeoff and a useful regularizing approximation in rolling-window settings with limited post-break samples, not as evidence that off-diagonal dynamics are absent. Between the strict diagonal and the unrestricted full matrix lies a broad middle ground that future work could exploit to retain the dominant cross-dynamics while discarding the rest: hierarchical shrinkage priors that pull the off-diagonal AR and MA elements toward zero unless the data support them \citep{katz2025sensitivity}, low-rank or factor parameterizations of the operators, and sparse estimation that selects only a small number of cross-effects. Such regularized specifications would target a better bias-variance balance than either extreme, while partially restoring the basis invariance that the strict diagonal form sacrifices.

Second, the logistic gate is monotonically increasing by construction, representing only permanent directionally consistent transitions. The reversibility simulation (Section~\ref{subsec:sim_reversibility}) quantifies the cost: coverage degrades from 79\% to 73\% as the recovery fraction increases from 0 to 80\%. The stay-length empirical application corroborates this, with the intervention's acceptable calibration weakening slightly at $h=6$ as forecasts extend into the reversion period. A natural extension, suggested by our reversibility design, is to bring the two-component gate $w_t^{\text{rev}}$ of Equation~\ref{eq:reversible_gate} into the likelihood rather than using it only to generate stress-test data, so the model can represent a shock followed by partial recovery. The obstacles are ones of identification more than implementation: the rise and recovery gates compete to explain the same post-break path, so their locations and speeds separate only weakly unless the reversion is pronounced, and the recovery fraction is poorly identified when the reversion is mild or the post-break sample is short. A workable version would need priors that order the two transitions ($t_{\text{rec}} > \tau$) and bound the recovery fraction, care with the resulting posterior multimodality, and more post-break data than a single rolling window typically provides; we leave this to future work.

Both the diagonal restriction and the monotone gate are fundamental modeling choices rather than incidental implementation details, and practitioners should assess both before applying the model. Finally, all empirical compositions in the present applications are strictly positive at the monthly aggregation used here; applications with zero components would require an explicit replacement step before Dirichlet modeling. Extensions to multiple breaks, unknown break dates, and hierarchical specifications across regions merit future investigation.

\section{Conclusion}
\label{sec:conclusion}

Our directional-shift Dirichlet ARMA model captures structural breaks in compositional time series through interpretable parameters (direction, amplitude, and timing) while preserving the Dirichlet likelihood and ARMA dynamics of the B-DARMA framework. The logistic gate provides smooth transitions that match empirical patterns of gradual adjustment. Simulations demonstrate accurate parameter recovery with nominal calibration across a wide range of transition speeds, and calibration degrades gracefully under partial reversibility. Two COVID-era empirical applications provide a precise characterization of when the model provides meaningful gains: the intervention model remains acceptably calibrated in both settings (79.6\% and 85.7\%), but its relative calibration advantage over the fixed effect is specific to the monotone lead-time application. A direct full-matrix sensitivity benchmark in the lead-time application further shows that naively relaxing the diagonal AR/MA restriction can materially worsen pooled out-of-sample point accuracy and calibration at each reported forecast horizon, supporting the diagonal specification as a practical regularizing approximation in rolling-window settings. The point accuracy advantage is conditional on the transition being roughly monotone. This characterization provides actionable guidance for practitioners choosing between break-aware specifications for compositional forecasting.

\bigskip

\section*{Code and Data Availability}
The primary datasets used in our study, the fee recognition lead-time and accommodation stay-length booking compositions, are not publicly available due to confidentiality constraints. However, the Stan code for the directional-shift Dirichlet Auto-Regressive Moving Average (DARMA) model is available for public access, as well as the R scripts for the simulation studies. It can be found at our GitHub repository: \href{https://github.com/harrisonekatz/gated_darma_project}{GitHub repository}.

\section*{Acknowledgement}
The author thanks Jess Needleman, Liz Medina, and Amy Ding for helpful discussions, Yuanyuan Cui, Peter Coles, and Adam Liss for championing the research, and the Editor-in-Chief, Associate Editors, and two anonymous reviewers for constructive feedback that improved the manuscript.

\bibliographystyle{chicago}
\bibliography{references}

\appendix

\section{Helmert Contrast Matrix}
\label{app:helmert}

We use a Helmert-style orthonormal contrast matrix for ILR transformation. For $C$ categories, $\mathbf{V} \in \R^{C \times (C-1)}$ with entries $V_{j,i} = 1/\sqrt{i(i+1)}$ for $j \leq i$, $V_{i+1,i} = -i/\sqrt{i(i+1)}$, and $V_{j,i} = 0$ for $j > i+1$. This satisfies $\mathbf{V}^\top \mathbf{V} = \mathbf{I}_{C-1}$.

\section{Geodesic Property and Basis Invariance}
\label{app:geodesic}

We establish that directional shifts in ILR space correspond to geodesics on the simplex under Aitchison geometry.

\begin{proposition}[Geodesic motion]
Let $\bm{\eta}_0 \in \R^{C-1}$ be an ILR coordinate and $\bm{v} \in \R^{C-1}$ a unit direction. The curve $\bm{\mu}(w) = \ilrinv(\bm{\eta}_0 + w\bm{v})$ for $w \in \R$ is a geodesic on $(\simplexC, d_A)$.
\end{proposition}

\begin{proof}
The ILR transformation is an isometry between $(\simplexC, d_A)$ and $(\R^{C-1}, \|\cdot\|_2)$ \citep{egozcue2003}. Geodesics in Euclidean space are straight lines. The curve $\bm{\eta}(w) = \bm{\eta}_0 + w\bm{v}$ is a straight line in $\R^{C-1}$. Since isometries preserve geodesics, $\bm{\mu}(w) = \ilrinv(\bm{\eta}(w))$ is a geodesic on the simplex.
\end{proof}

\begin{proposition}[Basis invariance]
Let $\mathbf{V}$ and $\tilde{\mathbf{V}}$ be two orthonormal ILR contrast matrices related by rotation $\mathbf{R} \in O(C-1)$, so $\tilde{\mathbf{V}} = \mathbf{V}\mathbf{R}$. If $(\bm{v}, \Delta)$ parameterizes a shift in $\mathbf{V}$-coordinates, then $(\tilde{\bm{v}}, \Delta)$ with $\tilde{\bm{v}} = \mathbf{R}^\top \bm{v}$ parameterizes the same simplex trajectory in $\tilde{\mathbf{V}}$-coordinates.
\end{proposition}

\begin{proof}
Let $\bm{\eta} = \mathbf{V}^\top \clr(Y)$ and $\tilde{\bm{\eta}} = \tilde{\mathbf{V}}^\top \clr(Y) = \mathbf{R}^\top \bm{\eta}$. The shifted coordinate in $\mathbf{V}$-space is $\bm{\eta} + \Delta \bm{v}$. In $\tilde{\mathbf{V}}$-space, this becomes $\mathbf{R}^\top(\bm{\eta} + \Delta \bm{v}) = \tilde{\bm{\eta}} + \Delta \mathbf{R}^\top \bm{v} = \tilde{\bm{\eta}} + \Delta \tilde{\bm{v}}$. Both map to the same simplex point under $\ilrinv$.
\end{proof}

These results justify reporting direction effects in CLR space (via $\bm{u} = \mathbf{V}\bm{v}$), which is basis-invariant since $\tilde{\mathbf{V}}\tilde{\bm{v}} = \mathbf{V}\mathbf{R}\mathbf{R}^\top\bm{v} = \mathbf{V}\bm{v}$.

\paragraph{Scope of invariance.} Note that this basis invariance applies to the intervention trajectory itself. When diagonal AR/MA dynamics are used (as in our implementation), the assumption of independent evolution across ILR coordinates is basis-dependent; rotating the ILR basis would change which directions evolve independently. The full fitted model is thus not basis-invariant, though the directional shift mechanism is.

\section{Prior Specification}
\label{app:priors}

Table~\ref{tab:priors} summarizes the prior distributions used in the empirical applications. The simulation study (Section~\ref{sec:simulation}) uses tighter priors for the intervention speed and location: $\tau \sim \N(\ell + 2, 3^2)$ and $\kappa \sim \text{LogNormal}(0, 0.5^2)$; all other priors are as listed below.

\begin{table}[htbp]
\centering
\caption{Prior distributions for model parameters.}
\label{tab:priors}
\begin{tabular}{lll}
\toprule
Parameter & Prior & Description \\
\midrule
\multicolumn{3}{l}{\emph{Baseline parameters}} \\
$b_d$ & $\N(0, 2.5^2)$ & ILR intercepts \\
$B_{d,k}$ & $\N(0, 1^2)$ & Covariate coefficients \\
$A_{d,d}$ & $\text{Uniform}(-0.99, 0.99)$ & AR diagonal \\
$\Theta_{d,d}$ & $\text{Uniform}(-0.99, 0.99)$ & MA diagonal \\
$\gamma_{\phi,k}$ & $\N(0, 1^2)$ & Concentration coefficients \\
\addlinespace
\multicolumn{3}{l}{\emph{Intervention parameters}} \\
$\Delta$ & $\N(0, 1.5^2)$ & Shift amplitude \\
$\tau$ & $\N(\ell + 2, 4^2)$ & Transition location parameter \\
$\kappa$ & $\text{LogNormal}(-0.5, 1^2)$ & Transition speed \\
$\bm{v}$ & Uniform on hemisphere & Direction (unit vector) \\
$\delta_\phi$ & $\N(0, 0.5^2)$ & Precision shift \\
\bottomrule
\end{tabular}
\end{table}

The hemisphere constraint on $\bm{v}$ resolves sign ambiguity by requiring $v_1 \geq 0$, implemented in Stan by constraining $v_1^{\text{raw}} > 0$ before normalization; see Section~\ref{sec:priors} for discussion.

\section{Evaluation Metrics}
\label{app:metrics}

\paragraph{Aitchison distance.} Point forecast accuracy is measured by $d_A(\hat{\bm{\mu}}_t, Y_t) = \|\clr(\hat{\bm{\mu}}_t) - \clr(Y_t)\|_2$, where $\hat{\bm{\mu}}_t$ is the posterior mean of the Dirichlet mean parameter.

\paragraph{Energy score.} The energy score generalizes CRPS to multivariate settings. We compute it in Aitchison geometry:
\[
\text{ES}(F, Y) = \E_F[d_A(\tilde{Y}, Y)] - \tfrac{1}{2}\E_F[d_A(\tilde{Y}, \tilde{Y}')]
\]
where $\tilde{Y}, \tilde{Y}'$ are independent draws from the posterior predictive $F$ and $Y$ is the observation. Lower is better.

\paragraph{Plug-in log score.} We report the log predictive density evaluated at posterior mean parameters: $\log p(Y_t \mid \hat{\bm{\mu}}_t, \hat{\lambda}_t) = \log \Dir(Y_t; \hat{\lambda}_t \hat{\bm{\mu}}_t)$. This plug-in approximation is computationally convenient and provides a useful comparison metric, though it does not account for parameter uncertainty. Higher is better.

\paragraph{Coverage.} We report marginal (componentwise) coverage: the proportion of observations falling within 80\% posterior predictive intervals, averaged across components. Nominal coverage is 80\%.

\paragraph{Mean absolute error.} Componentwise absolute error averaged across components and forecast cases: $\text{MAE} = \frac{1}{nC}\sum_{t,c}|\hat{\mu}_{t,c} - Y_{t,c}|$.

\section{Zero Handling}
\label{app:zeros}

The Dirichlet distribution requires strictly positive components. In our applications, all observed compositions have positive entries; no zero proportions occur in the aggregated monthly data. For applications where zeros may arise, standard approaches include additive replacement \citep{martinfernandez2003} or multiplicative replacement prior to analysis.

\section{Additional Regional Results: Lead-Time Application}
\label{app:regional}

Figures~\ref{fig:emea_heatmap}--\ref{fig:apac_fitted} present the exploratory data analysis heatmaps and fitted model plots for the EMEA, LATAM, and APAC regions, complementing the North American results shown in the main text.

\begin{figure}[htbp]
\centering
\begin{subfigure}[b]{0.48\textwidth}
\includegraphics[width=\textwidth]{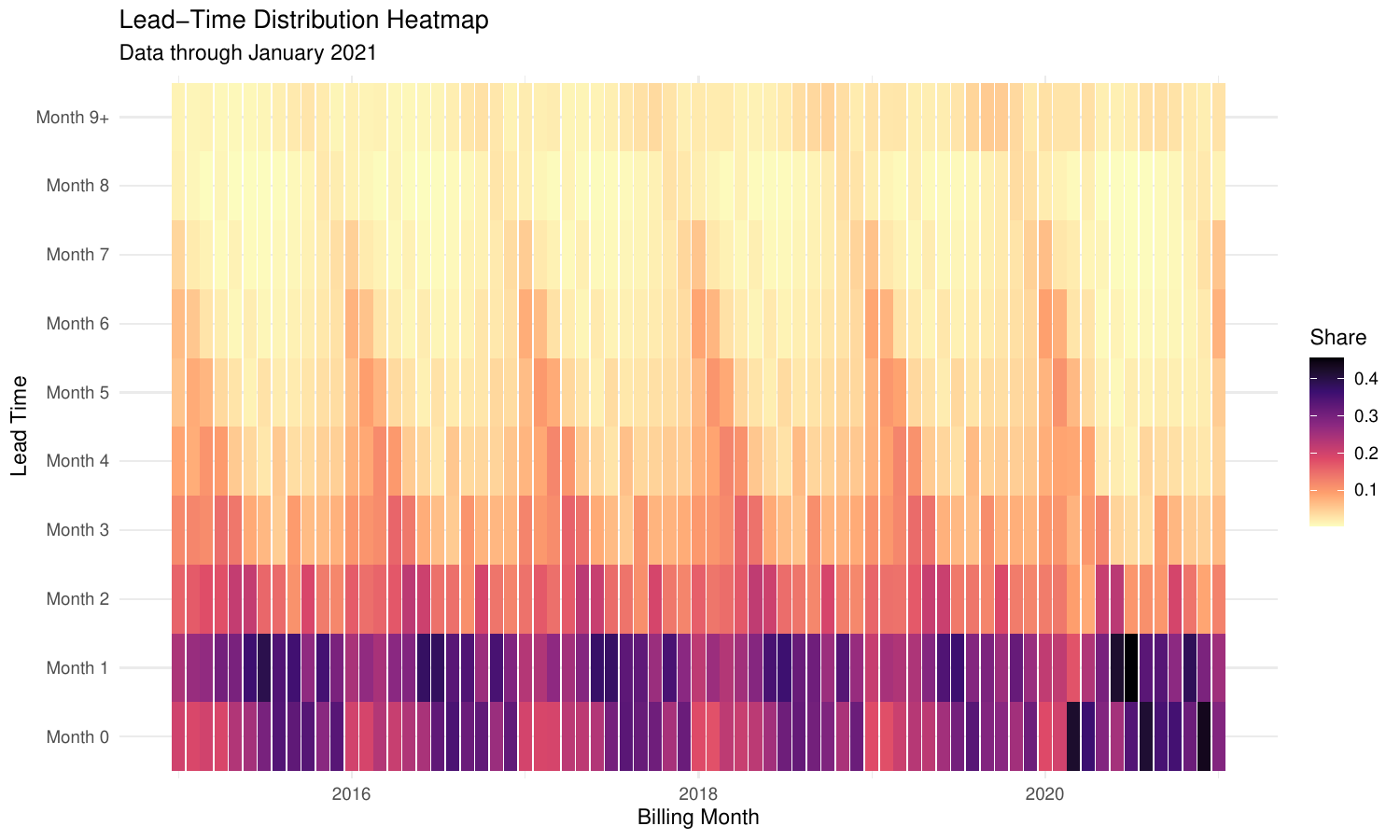}
\caption{EMEA}
\label{fig:emea_heatmap}
\end{subfigure}
\hfill
\begin{subfigure}[b]{0.48\textwidth}
\includegraphics[width=\textwidth]{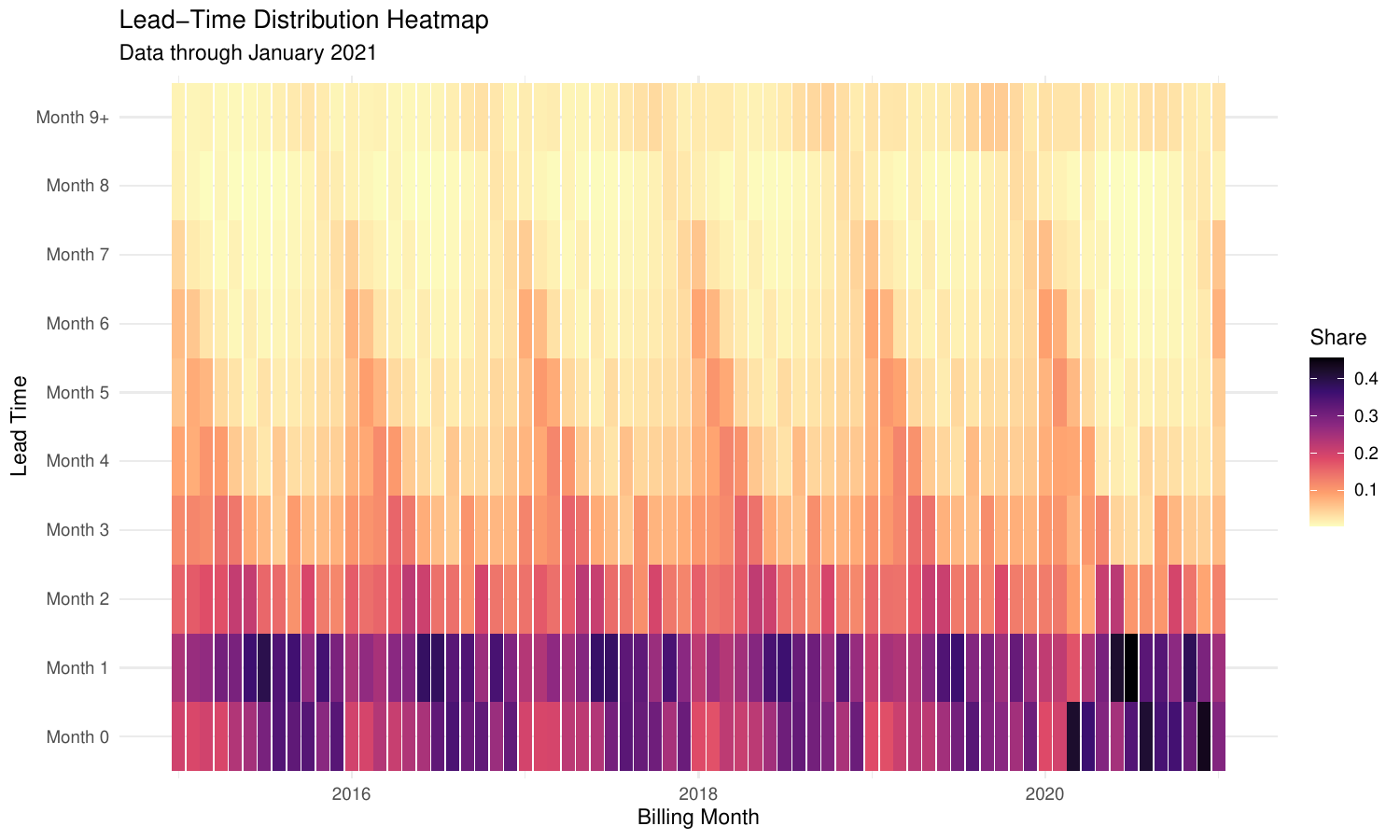}
\caption{LATAM}
\label{fig:latam_heatmap}
\end{subfigure}

\vspace{0.5cm}

\begin{subfigure}[b]{0.48\textwidth}
\includegraphics[width=\textwidth]{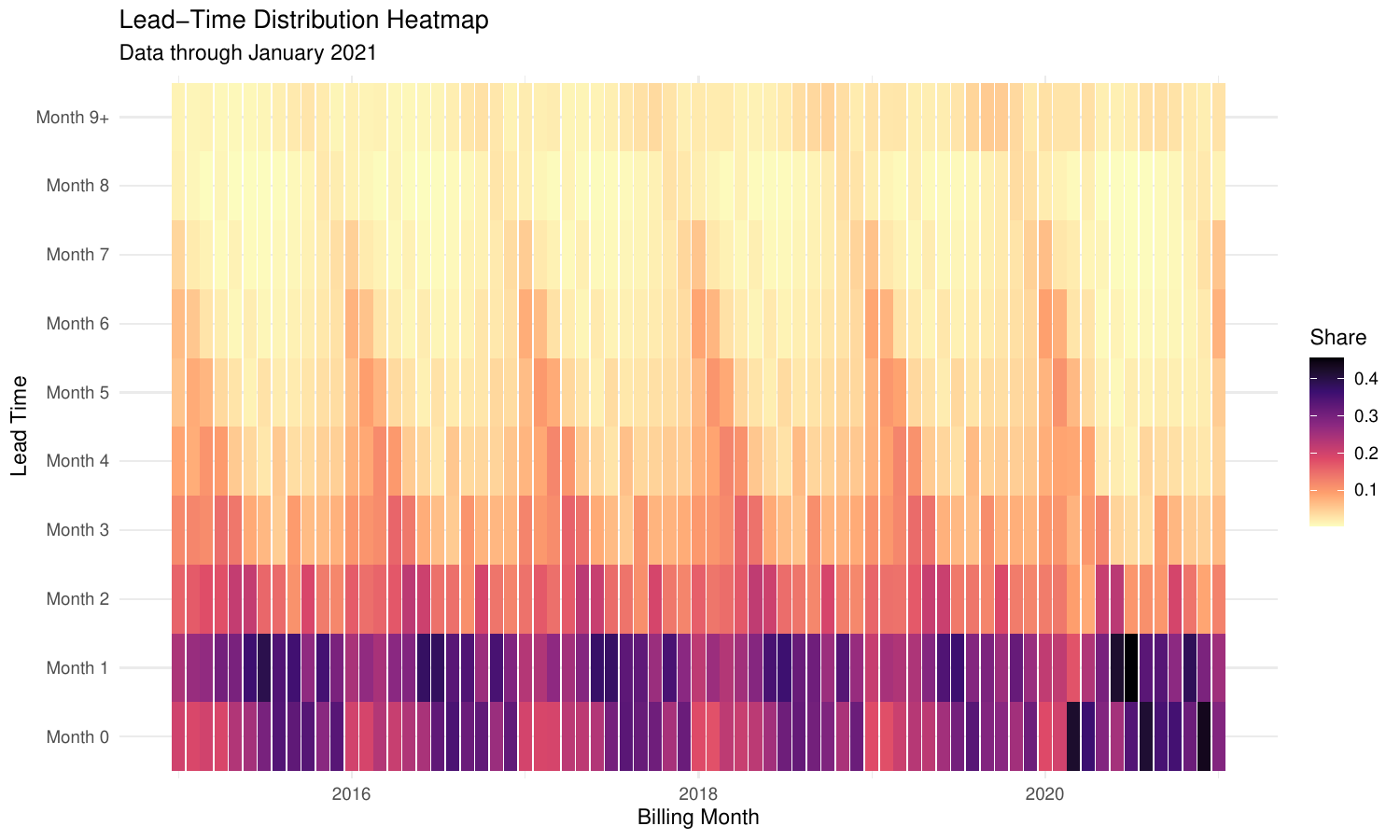}
\caption{APAC}
\label{fig:apac_heatmap}
\end{subfigure}
\caption{Lead-time distribution heatmaps for EMEA, LATAM, and APAC regions. All regions show the characteristic COVID-19 structural break in March 2020 with a shift toward shorter lead times, though the magnitude and timing of recovery varies across regions.}
\label{fig:regional_heatmaps}
\end{figure}

\begin{figure}[htbp]
\centering
\begin{subfigure}[b]{0.48\textwidth}
\includegraphics[width=\textwidth]{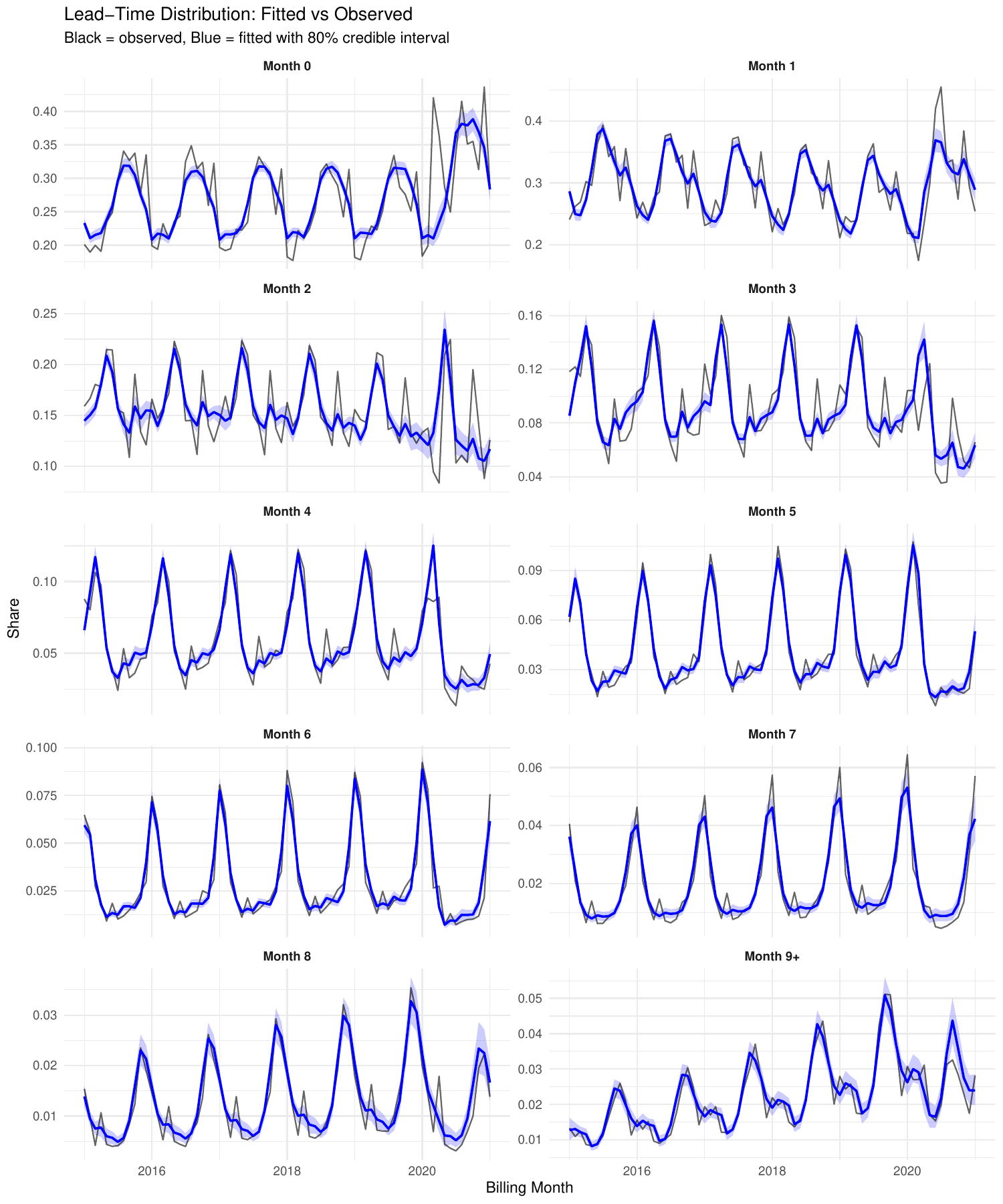}
\caption{EMEA}
\label{fig:emea_fitted}
\end{subfigure}
\hfill
\begin{subfigure}[b]{0.48\textwidth}
\includegraphics[width=\textwidth]{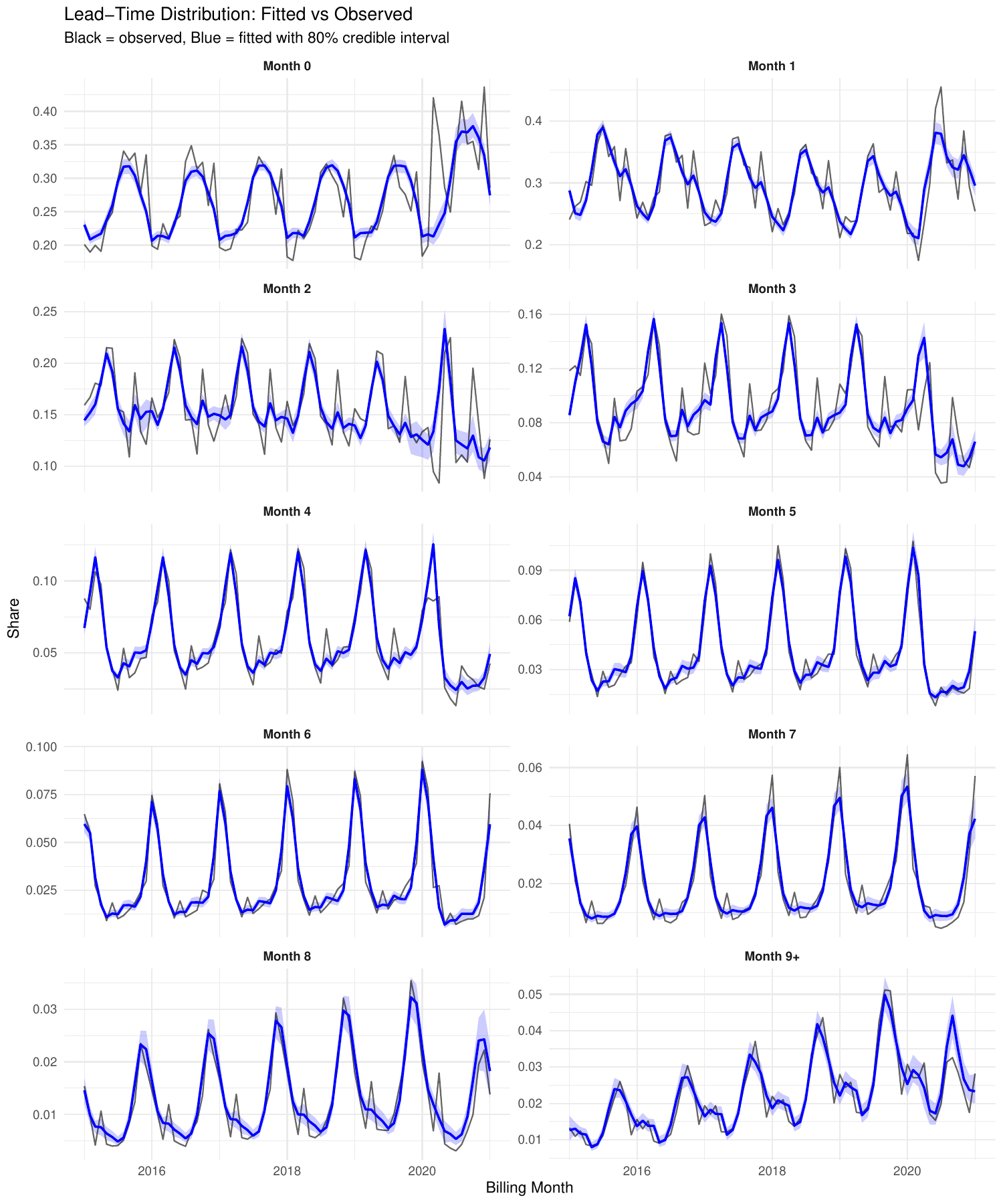}
\caption{LATAM}
\label{fig:latam_fitted}
\end{subfigure}

\vspace{0.5cm}

\begin{subfigure}[b]{0.48\textwidth}
\includegraphics[width=\textwidth]{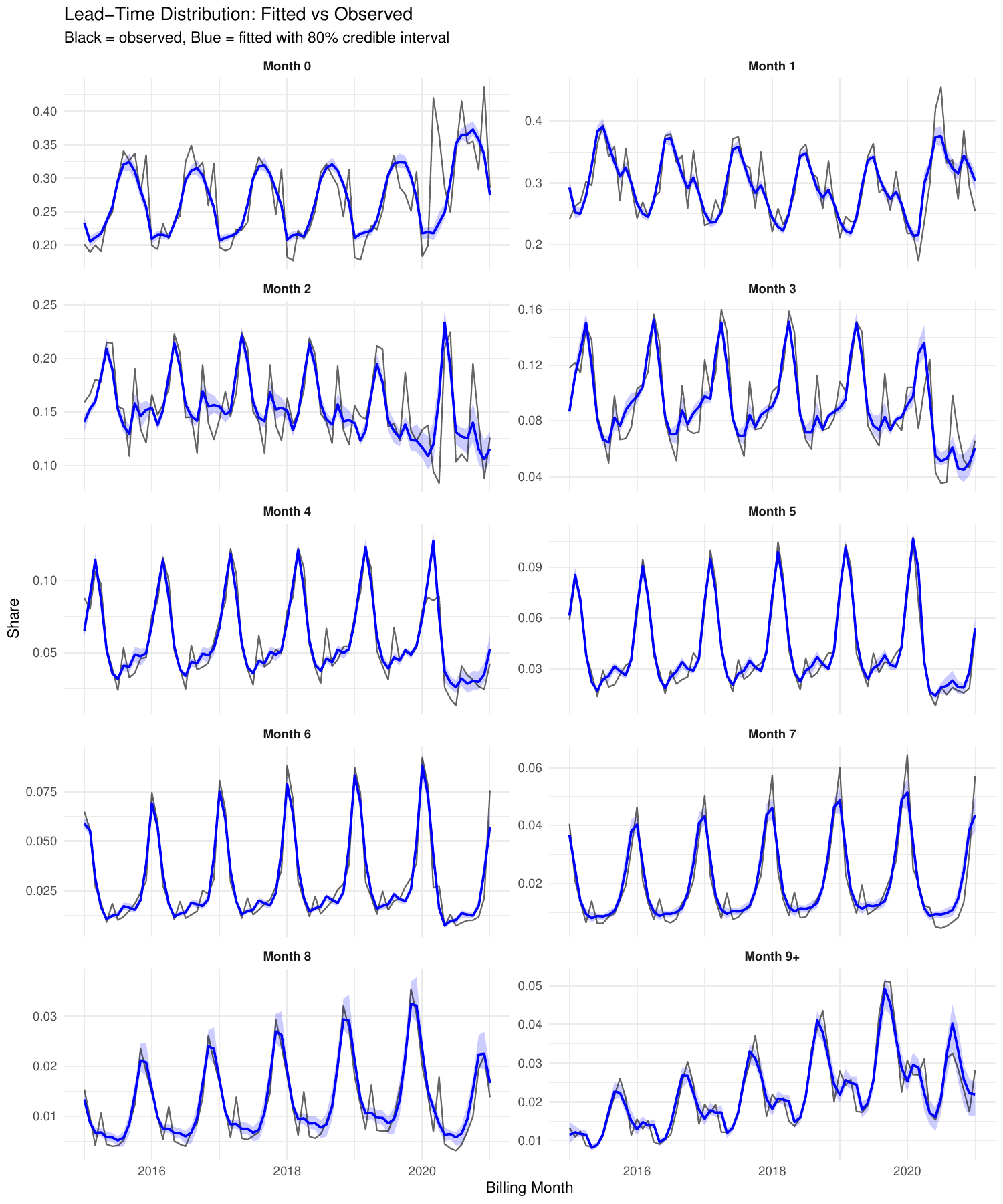}
\caption{APAC}
\label{fig:apac_fitted}
\end{subfigure}
\caption{Intervention model fitted values (blue) versus observed lead-time shares (black) for EMEA, LATAM, and APAC regions. Shaded bands show 80\% credible intervals. The model captures region-specific seasonal patterns and the COVID-19 structural shift across all markets.}
\label{fig:regional_fitted}
\end{figure}

\section{Stay-Length Application: APAC Results}
\label{app:los_apac}

Table~\ref{tab:los_apac} reports 1-step-ahead results for APAC, excluded from the main stay-length analysis due to the absence of a meaningful COVID-19 structural break in the stay-length composition. The intervention model achieves noticeably better point accuracy than the other two (Aitchison 0.203 vs.\ 0.242 for fixed effect and 0.261 for baseline) despite the absence of an obvious COVID arc in APAC, suggesting some residual directional signal in the composition. Coverage is similar across models (85--90\%), reflecting broad uncertainty about a series without a strong structural break rather than model-specific calibration properties.

\begin{table}[htbp]
\centering
\caption{Stay-length application: APAC 1-step-ahead results (appendix only). APAC is excluded from pooled main results due to minimal COVID signal.}
\label{tab:los_apac}
\begin{tabular}{lccccc}
\toprule
Model & Aitchison & Energy & Plug-in & MAE & Coverage \\
 & Distance & Score & Log Score & & (80\%) \\
\midrule
Baseline     & 0.261 & 0.253 & 13.8 & 0.0082 & \textbf{85.7\%} \\
Fixed Effect & 0.242 & 0.230 & 18.7 & 0.0055 & \textbf{85.7\%} \\
Intervention & \textbf{0.203} & \textbf{0.225} & \textbf{22.2} & \textbf{0.0052} & 89.8\% \\
\bottomrule
\end{tabular}
\end{table}

\end{document}